\let\oldvec\vec % Store \vec in \oldvec
\documentclass{llncs}
\let\vec\oldvec % Restore \vec from \oldvec

\usepackage{amsmath,amssymb}
\usepackage{color}
\usepackage{mathtools}
\usepackage{tikz}
\usepackage{url}
\usepackage{wrapfig}
\usepackage{enumerate}
\usepackage{cite}

\pgfdeclarelayer{background}
\pgfdeclarelayer{foreground}
\pgfsetlayers{background,main,foreground}

\usetikzlibrary{positioning,chains,fit,shapes,calc}
\usepackage{tikz}
\usetikzlibrary{arrows,automata,positioning}
\usepackage[latin1]{inputenc}

\newcommand{\N}{\mathbb{N}}
\newcommand{\R}{\mathbb{R}}

\newcommand{\p}{\ensuremath{\mathsf{P}}}
\newcommand{\np}{\ensuremath{\mathsf{NP}}}

\newcommand{\sset}[2]{\{#1 \; | \; #2\}}

\newcommand{\best}{best \deviation{} problem} %Use in text as \best{} 
\newcommand{\ratio}{Deviation Ratio} %Use in text as \ratio{} 
\newcommand{\dr}{\ensuremath{\text{DR}}}

\newcommand{\deviation}{deviation} %Use in text as \deviation{} 
\newcommand{\deviations}{deviations} %Use in text as \deviations{}
\newcommand{\deviated}{deviated} %Use in text as \deviated{}
\definecolor{myblue}{RGB}{80,80,160}
\definecolor{mygreen}{RGB}{80,160,80}

\makeatletter
\spn@wtheorem{assumption}{Assumption}{\itshape}{\rmfamily}
\makeatother

\usepackage{ifthen}
% ...
\newenvironment{rtheorem}[3][]{%
\noindent\ifthenelse{\equal{#1}{}}{\bf #2 #3.}{\bf #2 #3 (#1)}%
\begin{it}}{\end{it}}

\title{The Impact of Worst-Case Deviations in Non-Atomic Network Routing Games}
\author{Pieter Kleer\inst{1} \and Guido Sch\"afer\inst{1,2}}
\authorrunning{P. Kleer \and G.\ Sch\"afer} % a short form should be given in case it is too long for the running head

\institute{%
  Centrum Wiskunde \& Informatica (CWI), Networks and Optimization Group, Amsterdam, The Netherlands
  \and
  Vrije Universiteit Amsterdam, Department of Econometrics and Operations Research, Amsterdam, The Netherlands. \\ 
  \texttt{kleer@cwi.nl, schaefer@cwi.nl}
}

\begin{document}
\sloppy

\maketitle

\begin{abstract}
We introduce a unifying model to study the impact of worst-case latency deviations in non-atomic selfish routing games. In our model, latencies are subject to (bounded) deviations which are taken into account by the players. The quality deterioration caused by such deviations is assessed by the \emph{\ratio}, i.e., the worst case ratio of the cost of a Nash flow with respect to deviated latencies and the cost of a Nash flow with respect to the unaltered latencies. This notion is inspired by the \textit{Price of Risk Aversion} recently studied by Nikolova and Stier-Moses \cite{Nikolova2015}. Here we generalize their model and results. 
In particular, we derive tight bounds on the \ratio\ for multi-commodity instances with a common source and arbitrary non-negative and non-decreasing latency functions. These bounds exhibit a linear dependency on the size of the network (besides other parameters). In contrast, we show that for general multi-commodity networks an exponential dependency is inevitable. We also improve recent smoothness results to bound the Price of Risk Aversion. 
\end{abstract}


\section{Introduction}
\label{sec:contr}
\label{sec:relwork}

In the classical selfish routing game introduced by Wardrop \cite{Wardrop1952}, there is an (infinitely) large population of (non-atomic) players who selfishly choose minimum latency paths in a network with flow-dependent latency functions. An assumption that is made in this model is that the latency functions are given deterministically. Although being a meaningful abstraction (which also facilitates the analysis of such games), this assumption is overly simplistic in situations where latencies are subject to \deviation s which are taken into account by the players. 

In this paper, we study how much the quality of a Nash flow deteriorates in the worst case under (bounded) \deviation s of the latency functions. More precisely, given an instance of the selfish routing game with latency functions $(l_a)_{a \in A}$ on the arcs, we define the \emph{\ratio{} (DR)} as the worst case ratio $C(f^{\delta})/C(f^0)$ of a Nash flow $f^{\delta}$ with respect to \deviated\ latency functions $(l_a + \delta_a)_{a \in A}$, where $(\delta_a)_{a \in A}$ are arbitrary \deviation\ functions from a feasible set, and a Nash flow $f^0$ with respect to the unaltered latency functions $(l_a)_{a \in A}$. Here the social cost function $C$ refers to the total average latency (without the deviations). 
Our motivation for studying this social cost function is that a central designer usually cares about the long-term performance of the system (accounting for the average latency or pollution). On the other hand, the players typically do not know the exact latencies and use estimates or include ``safety margins'' in their planning. Similar viewpoints are adopted in \cite{Meir2015,Nikolova2015}.

%The main motivation for not including the deviations in the social cost is that we are interested in the setting where players are, citing Meir and Parkes \cite{Meir2015}, ``playing the wrong game". For example (see \cite{Meir2015}) players might not know the true latencies in the network, or include overly pessimistic ``safety margins'' in their travel time. See also Nikolova and Stier-Moses \cite{Nikolova2015} for motivation in the Price of Risk Aversion model. 

In order to model bounded deviations, we extend an idea previously put forward by Bonifaci, Salek and Sch\"afer \cite{Bonifaci2011} in the context of the  \textit{restricted network toll problem}: We assume that for every arc $a \in A$ we are given lower and upper bound restrictions $\theta^{\min}_a$ and $\theta^{\max}_a$, respectively, and call a \deviation\ $\delta_a$ \emph{feasible} if $\theta^{\min}_a(x) \le \delta_a(x) \le \theta^{\max}_a(x)$ for all $x \ge 0$. 

%In particular, we are interested in the \textit{\ratio{}}, the ratio between a \textit{\deviated{}} selfish flow, for some \deviation{} from a fixed restricted set, and the \textit{\classical{}} selfish flow, in order to measure the effect of the additional factors, but not that of the \classical{} selfishness (as in Wardrop's model \cite{Wardrop1952}). 

Our notion of the \ratio{} is inspired by and builds upon the \textit{Price of Risk Aversion (PRA)} recently introduced by Nikolova and Stier-Moses \cite{Nikolova2015}. The authors investigate selfish routing games with uncertain latencies by considering deviations of the form $\delta_a = \gamma v_a$, where $\gamma \geq 0$ is the risk-aversion of the players and $v_a$ is the variance of some random variable with mean zero.
% (modeling random deviations from the expected latency $l_a$). 
They derive upper bounds on the Price of Risk Aversion for single-commodity networks with arbitrary non-negative and non-decreasing latency functions if the \emph{variance-to-mean-ratio} $v_a/l_a$ of every arc $a \in A$ is bounded by some constant $\kappa \ge 0$. 
It is not hard to see that their model is a special case of our model if we choose $\theta^{\min}_a = 0$ and $\theta^{\max}_a = \gamma\kappa l_a$ (see Section~\ref{sec:pre} for more details). 

%In general, we believe that it is meaningful to study the Deviation Ratio of other games. 

%If the  \classical{} selfishness is also taken into account, a natural ratio to consider is that of the \deviated{} selfish flow and a socially optimal flow, which, in this setting, corresponds to the \textit{Biased Price of Anarchy} \cite{Meir2015} introduced by Meir and Parkes \cite{Meir2015}. 

%In this context, we distinguish between deviations that are imposed by a central designer or authority (\emph{exogenous deviations}), e.g., network tolls, and those which are inherent in the system (\emph{endogenous deviations}), e.g., stochastic latencies. 

%A classical example of the latter case is the work by Beckmann, McGuire and Winsten \cite{Beckmann1956}. The authors prove that so-called \textit{marginal tolls} impose a Nash flow (or Wardrop flow) that is optimal, i.e., minimizes the total average latency. A drawback of these marginal tolls is that they can be very large. Bonifaci, Salek and Sch\"afer \cite{Bonifaci2011} therefore introduce the \textit{restricted network toll} model, where the tolls can only be chosen from a restricted range. 

%In the latter case, a central authority might be interested in how much the performance of the network can deteriorate due to uncontrollable factors (on top of selfishness). We therefore introduce the worst-case equivalent of the restricted toll model, that is, which tolls (or \textit{\deviations{}}) from the restricted set would give the highest deterioration of the original selfish flow. 

%\bigskip
%\noindent
%\emph{Our contributions.} 
\paragraph{Our contributions.}
The main contributions presented in this paper are as follows: 

%\gsrem{slightly revised description of our contributions}

\smallskip
\noindent
\emph{1. Upper bounds:}
We derive a general upper bound on the \ratio\ for multi-commodity networks with a common source and arbitrary non-negative and non-decreasing latency functions (Theorem~\ref{lem:main_bound}).

In order to prove this upper bound, we first generalize a result by Bonifaci et al. \cite{Bonifaci2011} characterizing the inducibility of a fixed flow by $\delta$-\deviations\ to multi-commodity networks with a common source (Theorem~\ref{thm:induc}). This characterization naturally gives rise to the concept of an \emph{alternating path}, which plays a crucial role in the work by Nikolova and Stier-Moses \cite{Nikolova2015} and was first used by Lin, Roughgarden, Tardos and Walkover \cite{Lin2011} in the context of the \emph{network design problem}.

We then specialize our bound to the case of so-called \emph{$(\alpha, \beta)$-deviations}, where $\theta^{\min}_a = \alpha l_a$ and $\theta^{\max}_a = \beta l_a$ with $-1 < \alpha \le 0 \le \beta$. We prove that the Deviation Ratio is at most $1 + (\beta - \alpha) / (1 + \alpha) \lceil (n-1)/2 \rceil r$, where $n$ is the number of nodes of the network and $r$ is the sum of the demands of the commodities (Theorem~\ref{lem:main_bound}).
In particular, this reveals that the \ratio\ depends linearly on the size of the underlying network (among other parameters). 

By using this result, we obtain a bound on the Price of Risk Aversion (Theorem~\ref{thm:pra_model_results}) which generalizes the one in 
%of Nikolova and Stier-Moses 
\cite{Nikolova2015} in two ways: (i) it holds for multi-commodity networks with a common source and (ii) it allows for negative risk-aversion parameters (i.e., capturing risk-taking players as well). Further, we show that our result can be used to bound the relative error in social cost incurred by small latency perturbations (Theorem~\ref{thm:stability}), which is of independent interest. 

\smallskip
\noindent
\emph{2. Lower bounds:} 
We prove that our bound on the Deviation Ratio for $(\alpha, \beta)$-deviations is best possible. More specifically, for single-commodity networks we show that our bound is tight in all its parameters. Our lower bound construction holds for arbitrary $n \in \N$ and is based on the \textit{generalized Braess graph} \cite{Roughgarden2006} (Example \ref{exmp:tight_pra}).
In particular, this complements a recent result by Lianeas, Nikolova and Stier-Moses \cite{Lianeas2015} who show that their bound on the Price of Risk Aversion is tight for single-commodity networks with $n = 2^j$ nodes for all $j \in \N$. 
%The latency functions and \deviations{} that we use are also similar to those in \cite{Lianeas2015}, but our choice of graph topology gives rise to, in our opinion, a more direct and simpler proof.  

Further, for multi-commodity networks with a common source we show that our bound is tight in all parameters if $n$ is odd, while a small gap remains if $n$ is even (Theorem \ref{thm:lower_multi_odd}).
Finally, for general multi-commodity graphs we establish a lower bound showing that the \ratio{} can be exponential in $n$ (Theorem \ref{thm:multi_general}). In particular, this shows that there is an exponential gap between the cases of multi-commodity networks with and without a common source. In our proof, we adapt a graph structure used by Lin, Roughgarden, Tardos and Walkover \cite{Lin2011} in their lower bound construction for the network design problem on multi-commodity networks (see also \cite{Roughgarden2006}).

\smallskip
\noindent
\emph{3. Smoothness bounds:}
We improve (and slightly generalize) recent smoothness bounds on the Price of Risk Aversion given by Meir and Parkes \cite{Meir2015} and independently by Lianeas et al. \cite{Lianeas2015}. 
In particular, we derive tight bounds for the \textit{Biased Price of Anarchy (BPoA)} \cite{Meir2015}, i.e., the ratio between the cost of a \deviated\ Nash flow and the cost of a social optimum, for \emph{arbitrary} $(0,\beta)$-deviations (Theorem~\ref{thm:smooth-BPA}).\footnote{We remark that for certain types of $(0,\beta)$-deviations, e.g., scaled marginal tolls, better bounds can be obtained; see the section ``Related notions'' in Section 2 for relevant literature.}
Note that the Biased Price of Anarchy yields an upper bound on the Deviation Ratio/Price of Risk Aversion. 
We also derive smoothness results for general path deviations (which are not representable by arc deviations). As a result, we obtain bounds on the Price of Risk Aversion (Theorem~\ref{thm:smooth_general}) under the non-linear \emph{mean-std} model \cite{Lianeas2015, Nikolova2015} (see Section \ref{sec:pre}).

\bigskip
\noindent
It is interesting to note that the smoothness bounds on the Biased Price of Anarchy \cite{Meir2015} and the Price of Risk Aversion \cite{Lianeas2015} 
%derived by Meir and Parkes \cite{Meir2015} and Lianeas et al. \cite{Lianeas2015}
are independent of the network structure (but dependent on the class of latency functions). In contrast, the bound on the \ratio{} depends on certain parameters of the network.% (like the number of nodes $n$ and the total demand $r$).
\footnote{For example, there are parallel-arc networks for which the Biased Price of Anarchy is unbounded, whereas the \ratio{} is a constant.}

Our results answer a question posed in the work by Nikolova and Stier-Moses \cite{Nikolova2015} regarding possible relations between their Price of Risk Aversion model \cite{Nikolova2015}, the restricted network toll problem \cite{Bonifaci2011}, and the network design problem\cite{Roughgarden2006}. 
In particular, our results also show that the analysis in \cite{Nikolova2015} is not inherent to the used variance function, but rather depends on the restrictions imposed on the feasible deviations.

%\gsrem{significantly shortened the related work section and moved some material from the preliminaries to here}
%\medskip
%\noindent
%\emph{Related work.}
\paragraph{Related work.}
The modeling and studying of uncertainties in routing games has received a lot of attention in recent years. An extensive survey on this topic is given by Cominetti \cite{Cominetti2015}. 

As mentioned above, our investigations are inspired by the study of the Price of Risk Aversion by Nikolova and Stier-Moses \cite{Nikolova2015}. They prove that for single-commodity instances with non-negative and non-decreasing latency functions the Price of Risk Aversion is at most $1 + \gamma \kappa \lceil (n-1)/2 \rceil$. We elaborate in more detail on the connections to their work in Section~\ref{sec:pre}. 

%\gsrem{For single-commodity instances the demand $r$ does not enter the bound. With common-source multi-commodity this is dealt with by normalizing $r_1 = 1$. Without normalization it still works, but then all demands get scaled by $r_1$ in the bound. See also comment on this on page 20.}

%The most relevant works on perceived latencies are given after the description of the model. These also include some algorithmic results. %Roughly speaking, we generalize a result from \cite{Bonifaci2011}, that we use to extend upper bounds obtained in\cite{Nikolova2015}, for which we can provide matching lower bounds by modifying instances from \cite{Roughgarden2006}. \\
There are several papers that study the problem of imposing tolls (which can be viewed as latency deviations, see Section \ref{sec:pre} for more details) on the arcs of a network to reduce the cost of the resulting Nash flow. Conceptually, our model is related to the \emph{restricted network toll problem} by Bonifaci et al.~\cite{Bonifaci2011}. The authors study the problem of computing non-negative tolls that have to obey some upper bound restrictions $(\theta_a)_{a \in A}$ such that the cost of the resulting Nash flow is minimized. This is tantamount to computing best-case deviations in our model with $\theta_a^{\min} = 0$ and $\theta_a^{\max} = \theta_a$. In contrast, our focus here is on worst-case deviations. As a side result, we prove that computing such worst-case \deviations{} is NP-hard, even for single-commodity instances with linear latencies (Theorem~\ref{thm:np-hard_max}).

Roughgarden \cite{Roughgarden2006} studies the \emph{network design problem} of finding a subnetwork that minimizes the latency of all flow-carrying paths of the resulting Nash flow. He proves that the \textit{trivial algorithm} (which simply returns the original network) gives an $\lfloor n/2\rfloor$-approximation algorithm for single-commodity networks and that this is best possible (unless $\p = \np$).
%The instances that we use in Example \ref{exmp:tight_pra} and Theorem \ref{thm:lower_multi_odd} are based on the instances used to prove the tightness of the $\lfloor n/2\rfloor$-approximation algorithm. 
Later, Lin et al. \cite{Lin2011} show that this algorithm can be exponentially bad for multi-commodity networks. The instances that we use in our lower bound constructions are based on the ones used in \cite{Roughgarden2006,Lin2011}. 

Meir and Parkes \cite{Meir2015} and independently Lineas et al. \cite{Lianeas2015} show that for non-atomic network routing games with \emph{$(1,\mu)$-smooth}%
\footnote{Meir and Parkes \cite{Meir2015} define a function $l$ to be \emph{$(1,\mu)$-smooth} if $x l(y) \leq \mu yl(y) +xl(x)$ for all $x,y \geq 0$ (which is slightly different from Roughgarden's original smoothness definition \cite{Roughgarden2015}). Lineas et al. \cite{Lianeas2015} only require \emph{local smoothness} where $y$ is taken fixed.} 
latency functions it holds that $\text{PRA} \leq  \text{BPoA} \leq (1 + \gamma \kappa)/(1- \mu)$. 
An advantage of such bounds is that they hold for general multi-commodity instances (but depend on the class of latency functions). These bounds stand in contrast to the \textit{topological} bounds obtained here and by Nikolova and Stier-Moses \cite{Nikolova2015} which hold for arbitrary non-negative and non-decreasing latency functions.

\section{Preliminaries}
\label{sec:pre}

\paragraph{Bounded deviation model.}
%\gsrem{rephrased slightly}
Let $\mathcal{I} = (G = (V,A),(l_a)_{a \in A},(s_i,t_i)_{i \in [k]}, (r_i)_{i \in [k]})$ be an instance of a non-atomic network routing game. Here, $G = (V,A)$ is a directed graph with node set $V$ and arc set $A \subseteq V \times V$, where each arc $a \in A$ has a non-negative, non-decreasing and continuous latency function $l_a: \mathbb{R}_{\ge 0} \rightarrow \mathbb{R}_{\ge 0}$. Each commodity $i \in [k]$ is associated with a source-destination pair $(s_i, t_i)$ and has a demand of $r_i \in \mathbb{R}_{>0}$. We assume that $t_i \neq t_j$ if $i \neq j$ for $i,j \in [k]$. If all commodities share a common source node, i.e., $s_i = s_j = s$ for all $i, j \in [k]$, we call $\mathcal{I}$ a \textit{common source multi-commodity instance (with source $s$)}. We assume without loss of generality that $1 = r_1 \leq r_2 \leq \dots \leq r_k$ and define $r = \sum_{i \in [k]} r_i$.

We denote by $\mathcal{P}_i$ the set of all simple $(s_i,t_i)$-paths of commodity $i \in [k]$ in $G$, and we define $\mathcal{P} = \cup_{i \in [k]} \mathcal{P}_i$. An outcome of the game is a feasible flow $f: \mathcal{P} \rightarrow \R_{\geq 0}$, i.e., $\sum_{P \in \mathcal{P}_i} f_P = r_i$ for every $i \in [k]$. Given a flow $f = (f^i)_{i \in [k]}$, we use $f^i_a$ to denote the total flow on arc $a \in A$ of commodity $i \in [k]$, i.e., $f_a^i = \sum_{P \in \mathcal{P}_i : a \in P} f_P$. The total flow on arc $a \in A$ is defined as $f_a = \sum_{i \in [k]} f_a^i$. 
The latency of a path $P \in \mathcal{P}$ with respect to $f$ is  defined as $l_P(f) := \sum_{a \in P} l_a(f_a)$. The \textit{social cost} $C(f)$ of a flow $f$ is given by its total average latency, i.e., $C(f) = \sum_{P \in \mathcal{P}} f_Pl_P(f) = \sum_{a \in A} f_a l_a(f_a)$.
A flow that minimizes $C(\cdot)$ is called \textit{(socially) optimal}.
We use $A^+_i = \{a \in A : f_a^i > 0\}$ to refer to the support of $f^i$ for commodity $i \in [k]$ and define $A^+ = \cup_{i \in [k]} A^+_i$ as the support of $f$.

For every arc $a \in A$, we have a continuous function $\delta_a : \R_{\geq 0} \rightarrow \R$ modeling the \textit{\deviation{}} on arc $a$, and we write $\delta = (\delta_a)_{a \in A}$. We define the deviation of a path $P \in \mathcal{P}$ as $\delta_P(f) = \sum_{a \in P} \delta_a(f_a)$. The \textit{\deviated{} latency} on arc $a \in A$ is given by $q_a(f_a) = l_a(f_a) + \delta_a(f_a)$;
%\footnote{We do not include a sensitivity parameter $\gamma_j$ for the deviations, i.e., $q_a(f) = l_a(f) + \gamma_j \delta_a(f)$, since all our results only apply to the homogeneous case.}
similarly, the \deviated{} latency on path $P \in \mathcal{P}$ is given by $q_{P}(f) = l_P(f) + \delta_P(f)$. We say that $f$ is $\delta$\textit{-inducible} if and only if it is a \textit{Wardrop flow} (or \textit{Nash flow}) with respect to $l + \delta$, i.e., 
\begin{equation}
\forall i \in [k], \forall P \in \mathcal{P}_i, f_P > 0: \ \ \ \ \ \ q_{P}(f) \leq q_{P'}(f)  \ \ \forall P' \in \mathcal{P}_i.
\label{eq:nash}
\end{equation}
If $f$ is $\delta$-inducible, we also write $f = f^\delta$. 
Note that a Nash flow $f$ for the unaltered latencies $(l_a)_{a \in A}$ is $0$-inducible, i.e., $f = f^0$.

Let $\theta^{\min} = (\theta^{\min}_a)_{a \in A}$ and $\theta^{\max} = (\theta^{\max}_a)_{a \in A}$ be given continuous threshold functions satisfying $\theta_a^{\min}(x) \leq 0 \leq \theta_a^{\max}(x)$ for all $x \geq 0$ and $a \in A$, and let $\theta = (\theta^{\min},\theta^{\max})$. We define
$\Delta(\theta) = \sset{(\delta_a)_{a \in A}}{\forall a \in A : \theta_a^{\min}(x) \leq \delta_a(x) \leq \theta_a^{\max}(x), \ \forall x \geq 0}$ as the set of feasible \deviations{}.
Note that $0 \in \Delta(\theta)$ for all threshold functions $\theta^{\min}$ and $\theta^{\max}$. We say that $\delta \in \Delta(\theta)$ is a \textit{$\theta$-\deviation{}}. Furthermore, $f$ is \textit{$\theta$-inducible} if there exists a $\delta \in \Delta(\theta)$ such that $f$ is $\delta$-inducible. 
For $-1 < \alpha \leq 0 \leq \beta$, we call $\delta \in \Delta(\theta)$ an \textit{$(\alpha,\beta)$-\deviation{}} if $\theta^{\min} = \alpha l $ and $\theta^{\max} = \beta l$, and also write $\theta = (\alpha,\beta)$. 

%\gsrem{changed property to assumption here}
We make the following assumption throughout the paper: 
%We assume that the deviated latencies are always non-negative for any \deviation{} $\delta \in \Delta(\theta)$: 
\begin{assumption}
We assume that $l_a(x) + \theta^{\min}_a(x) \geq 0$ for all $x \geq 0$ and $a \in A$.
\label{ass:nonnegative}
\end{assumption}
The restrictions imposed on the deviations naturally give rise to the following two optimization problems. We emphasize that in both problems the social cost function $C(\cdot)$ only takes into account the latencies but not the \deviations.
\begin{enumerate}
\item \textit{Best Deviation Problem}: compute a \deviation{} $\delta \in \Delta(\theta)$ which minimizes $\inf\{C(f^{\delta}) : \delta \in \Delta(\theta)\}$. 
If $f^\delta$ is not unique, we assume that $C(f^\delta)$ refers to the social cost of the best Nash flow. 
 
\item \textit{Worst Deviation Problem}: compute a \deviation{} $\delta \in \Delta(\theta)$ which maximizes $\sup\{C(f^{\delta}) : \delta \in \Delta(\theta)\}$.
If $f^\delta$ is not unique, we assume that $C(f^\delta)$ refers to the social cost of the worst Nash flow. 
\end{enumerate}

%\gsrem{removed some parameters from the discussion here}
We (implicitly) assume that only \deviations{} $\delta$ are considered for which a Nash flow exists. We briefly elaborate on the existence when $\theta^{\min} = 0$ and $\theta^{\max}_a$ is  non-negative, non-decreasing and continuous for all $a \in A$. It is not hard to see that for a \deviated{} Nash flow  $f^\delta$ there exists some $0 \leq \lambda_a \leq 1$ for every arc $a \in A$ such that $\delta_a(f_a^\delta) = \lambda_a \theta^{\max}_a(f_a^\delta)$. In particular, this means that $\delta' \in \Delta(\theta)$ defined by $\delta_a' = \lambda_a \theta^{\max}_a$ also induces $f^\delta$. Therefore it is sufficient to consider \deviations{} of the form $\delta_a = \lambda_a \theta^{\max}_a$ where $0 \leq \lambda_a \leq 1$ for all $a \in A$. As a consequence, it follows that $q_a = l_a + \delta_a$ is a  non-negative, non-decreasing and continuous function for all $a \in A$. It is well-known that for these types of functions, the existence of a Nash flow is guaranteed (see, e.g., Nisan et al. \cite{Nisan2007}).

%\gsrem{removed equation numbering (not used too often) and capitalized problems (only here)}
%The restrictions imposed on the deviations naturally give rise to the following optimization problem. 

%\begin{enumerate}
%\item \textit{Best Deviation Problem}: compute a \deviation{} $\delta \in \Delta(\theta)$ which minimizes $\inf\{C(f^{\delta}) : \delta \in \Delta(\theta)\}$. 
%If $f^\delta$ is not unique, we assume that $C(f^\delta)$ represents the social cost of the best Nash flow. 
%\pkrem{Best Deviation Problem removed}
%\item 
%\textit{Worst Deviation Problem}: compute a \deviation{} $\delta \in \Delta(\theta)$ which maximizes $\sup\{C(f^{\delta}) : \delta \in \Delta(\theta)\}$.
%If $f^\delta$ is not unique, we assume that $C(f^\delta)$ represents the social cost of the worst Nash flow. 
%\end{enumerate}

%\gsrem{removed the mentioning of worst deviation problem here as this is not our main focus; please check whether it is used elsewhere}
%\gsrem{adapted the DR notion (depending on $\theta$ instead of specific $\delta$); please check}
\paragraph{Deviation Ratio.}
Given an instance $\mathcal{I}$ and threshold functions $\theta = (\theta^{\min}, \theta^{\max})$, we define the \emph{\ratio{}} $\dr(\mathcal{I},\theta) = \sup_{\delta \in \Delta(\theta)} C(f^\delta)/C(f^0)$ as the worst-case ratio of the cost of a $\theta$-inducible flow and the cost of a $0$-inducible flow. Intuitively, $\dr(\mathcal{I}, \theta)$ measures the worst-case deterioration of the social cost of a Nash flow due to (feasible) latency deviations. 
%Given a collection of instances $\mathcal{G}$ and threshold functions $\Theta$, we define $\text{DR}(\mathcal{G}, \Theta) = \sup_{(\mathcal{I},\theta) \in (\mathcal{G}, \Theta)} \text{DR}(\mathcal{I},\theta)$. 

%We emphasize that the social cost function $C$ is defined as above, i.e., with respect to the latencies (not taking into account the deviations). 
Note that for fixed deviations $\delta \in \Delta(\theta)$, there might be multiple Nash flows that are $\delta$-inducible. In this case, we adopt the convention that $C(f^\delta)$ refers to the social cost of the worst Nash flow that is $\delta$-inducible. 

Our main focus in this paper is on establishing (tight) bounds on the Deviation Ratio. As a side-result, we prove that the problem of determining worst-case deviations is \np-hard. 
%We first show that the \worst{}\ is NP-hard by a reduction DIRECTED HAMILTIONIAN PATH problem, which has been shown to be NP-hard by Garey and Johnson \cite{Garey1979}. 
\begin{theorem}
It is \np-hard to compute deviations $\delta \in \Delta(\theta)$ such that $C(f^\delta)$ is maximized, even for single-commodity networks with linear latencies.
%Given an instance $\mathcal{I}$, threshold functions $\theta$ and a parameter $K$, it is \np-complete to determine whether there exist deviations $\delta \in \Delta(\theta)$ such that $C(f^\delta) \ge K$, even for single-commodity networks with linear latencies.
\label{thm:np-hard_max}
\end{theorem}

\paragraph{Related notions.} 
% \cite{Nikolova2015, Lianeas2015, Nikolova2014}. 
The best \deviation{} problem is a direct generalization of the \textit{restricted network toll problem} introduced by Bonifaci et al.~\cite{Bonifaci2011}. We obtain this model for $\theta^{\min} = 0$. The \deviations{} are interpreted as non-negative tolls on the arcs. The objective minimized in \cite{Bonifaci2011} is measured against the social optimum, i.e., the authors are interested in the ratio $C(f^\delta)/C(f^*)$, where $f^*$ is an optimal flow for the instance $\mathcal{I}$. Also, our definition of $(0, \beta)$-\deviations{} is equivalent to the definition of \textit{$\beta$-restricted tolls} in \cite{Bonifaci2011}.

Hoefer et al. \cite{Hoefer2008} consider the \emph{taxing subnetwork problem}, which is a special case of the restricted network toll problem. Here only a designated subset of the arcs can be tolled, which is equivalent to $\theta^{\min}_a = 0$ and $\theta^{\max}_a \in \{0,\infty\}$ for all $a \in A$. They show that \best{} is  NP-complete, even for two commodities. To the best of our knowledge, the single-commodity case is still an open problem. On the positive side, Hoefer et al. \cite{Hoefer2008} and Bonifaci et al. \cite{Bonifaci2011} give polynomial time algorithms for parallel-arc networks, solving the best \deviation{} problem for their respective definitions of the threshold functions.  

Lastly, the work by Fotakis et al. \cite{Fotakis2015} can technically be seen as an (approximation) variant of the restricted toll model, in which the tolls are interpreted as risk-averse behavior of players. Here, we have $\theta^{\min}_a = 0$ and $\theta^{\max}_a = \gamma l_a$ for all $a \in A$. Furthermore, \deviations{} of the form $\delta_a(x) = \gamma_a l_a(x)$ are considered for $0 \leq \gamma_a \leq \gamma$ for all $a \in A$.

Beckmann et al. \cite{Beckmann1956} proved that the social optimum can be induced as a Nash flow using \textit{marginal tolls}, that is, by setting $\delta_a(x) = x \cdot l'_a(x)$, where $l'_a(x)$ is the derivative of $l_a(x)$ (assuming the existence of $l_a'$). In particular, if these tolls are feasible, i.e., $\delta \in \Delta(\theta)$, then $\delta$ is an optimal solution for the best \deviation{} problem. An extension of this setting, which has been studied intensively recently, is to consider perceived latencies of the form $l_a(x) + \rho \cdot xl'_a(x)$ for some parameter $\rho \in \R$, i.e., we take $\delta_a(x) = \rho \cdot xl'_a(x)$. This type of \deviation{} can be interpreted in many ways. If there exists a $\rho$ such that $(\rho \cdot xl'_a(x))_{a \in A} \in \Delta(\theta)$, then this \deviation{} gives an approximation for the \best{}. Results that are related to this are \cite{Chen2014, Chen2008, Christodoulou2011, Fotakis2015, Meir2014, Meir2015}.

Nikolova and Stier-Moses \cite{Nikolova2015} (see also \cite{Lianeas2015, Nikolova2014}) consider non-atomic network routing games with uncertain latencies.
% which are captured by our deviation model. 
Here the \deviations{} correspond to variances $(v_a)_{a \in A}$ of some random variable $\zeta_a$ (with expectation zero). 
%\gsrem{Please check! I changed this to path based formulation (as the arc-based perceived latency definition does not make sense if we consider the mean-std objective).}
The \emph{perceived latency} of a path $P \in \mathcal{P}$ with respect to a flow $f$ is then defined as $q^\gamma_P(f) = l_P(f) + \gamma v_P(f)$, where $\gamma \geq 0$ is a parameter representing the \emph{risk-aversion} of the players. 
They consider two different objectives as to how the deviation $v_P(f)$ of a path $P$ is defined: $v_P(f) = \sum_{a \in P} v_a(f_a)$, called the \textit{mean-var} objective, and $v_P(f) = (\sum_{a \in P} v_a(f_a))^{1/2}$, called the \textit{mean-std} objective. 
Note that for the mean-var objective there is an equivalent arc-based definition, where the perceived latency of every arc $a \in A$ is defined as $q^\gamma_a(f_a) = l_a(f_a) + \gamma v_a(f_a)$.
They define the \emph{Price of Risk Aversion} \cite{Nikolova2015}  as the worst-case ratio $C(x)/C(z)$, where $x$ is a \emph{risk-averse} Nash flow with respect to $q^\gamma = l + \gamma v$ and $z$ is a \emph{risk-neutral} Nash flow with respect to $l$.\footnote{The existence of a risk-averse Nash flow is proven in \cite{Nikolova2014}.}  
In their analysis, it is assumed that the \textit{variance-to-mean-ratio} of every arc $a \in A$ under the risk-averse flow $x$ is bounded by some constant $\kappa \ge 0$, i.e., $v_a(x_a) \leq \kappa \l_a(x_a)$ for all $a \in A$. Under this assumption, they prove that the Price of Risk Aversion $\text{PRA}(\mathcal{I}, \gamma, \kappa)$ of single-commodity instances $\mathcal{I}$ with non-negative and non-decreasing latency functions is at most $1 + \gamma \kappa \lceil (n-1)/2 \rceil$, where $n$ is the number of nodes. % and $r$ is the demand of the (single) commodity. 
%\gsrem{Removed $r$ from the bound}

We now elaborate on the relation to our \ratio. The main technical difference is that in \cite{Nikolova2015} the variance-to-mean ratio is only considered for the respective flow values $x_a$. 
%\gsrem{please check; there was a redundant $\gamma$ here I believe. YES, INDEED.} 
Note however that if we write for every $a \in A$, $v_a(x_a) = \lambda_a l_a(x_a)$ for some $0 \leq \lambda_a \leq \kappa$, then the deviation function $\delta_a(y) =  \gamma\lambda_a l_a(y)$ has the property that $x = f^{\delta}$ is $\delta$-inducible with $\delta \in \Delta(0,\gamma \kappa)$. It follows that for every instance $\mathcal{I}$ and parameters $\gamma$, $\kappa$, 
$\text{PRA}(\mathcal{I}, \gamma, \kappa) \le \dr(\mathcal{I}, (0, \gamma\kappa))$. 

%This is a direct generalization of the \textit{Price of Risk Aversion (PRA)}, introduced by Nikolova and Stier-Moses \cite{Nikolova2015}. 

Another related notion is the \emph{Biased Price of Anarchy (BPoA)} introduced by Meir and Parkes \cite{Meir2015}. Adapted to our setting, given an instance $\mathcal{I}$ and threshold functions $\theta$, the Biased Price of Anarchy is defined as $\text{BPoA}(\mathcal{I},\theta) = \sup_{\delta \in \Delta(\theta)} C(f^\delta)/C(f^*)$, where $f^*$ is a socially optimal flow. Note that because $C(f^*) \le C(f)$ for every feasible flow $f$, we have $\dr(\mathcal{I}, \theta) \le \text{BPoA}(\mathcal{I},\theta)$.

%Recently, Meir and Parkes \cite{Meir2015} and Lineas et al. \cite{Lianeas2015} have provided PRA-bounds for specific classes of latency functions. These bounds can be both superior and inferior to the topological bounds derived by Nikolova and Stier-Moses \cite{Nikolova2015}, depending on the class of latency functions and graph topologies considered.
%Meir and Parkes \cite{Meir2015} do this from a generalized smoothness framework, based on Roughgarden's framework in \cite{Roughgarden2015}.

%In our model, we could ask a similar question as in the network design problem: Given a \deviated{} Nash flow, on which arcs should we remove the \deviations{} in order to get a better (\deviated{}) Nash flow? The trivial algorithm would be to remove no \deviations{}. The \ratio{}\ then yields an approximation guarantee on this trivial algorithm. This might also explain the close relation to the results obtained in \cite{Roughgarden2006} and \cite{Lin2011}. 

\medskip
\noindent
Due to space limitations, some material is omitted from the main text and can be found in the appendix.

\section{Upper bounds on the \ratio{}}
\label{sec:upper-bounds}

We derive an upper bound on the Deviation Ratio. All results in this section hold for multi-commodity instances with a common source.

We first derive a characterization result for the inducibility of a given flow $f$. This generalizes the characterization in \cite{Bonifaci2011} to common source multi-commodity instances and negative deviations. 
We define an \emph{auxiliary graph} $\hat{G} = \hat{G}(f) = (V,\hat{A})$ with $\hat{A} = A \cup \bar{A}$, where $\bar{A} = \{ (v,u) : a = (u,v) \in A^+\}$. That is, $\hat{A}$ consists of the set of arcs in $A$, which we call \textit{forward} arcs, and the set $\bar{A}$ of arcs $(v, u)$ with $(u,v) \in A^+$, which we call \textit{reversed arcs}. 
Further, we define a cost function $c : \hat{A} \rightarrow \R$ as follows: 
\begin{equation}
c_a = \left\{ \begin{array}{rl} l_{(u,v)}(f_a) + \theta^{\max}_{(u,v)}(f_a) & \text{ for } a = (u,v) \in A \\
-l_{(u,v)}(f_a) - \theta^{\min}_{(u,v)}(f_a) & \text{ for } a = (v,u) \in \bar{A}.
\end{array}\right.
\label{def:cost}
\end{equation}

\begin{theorem}
%Let $\mathcal{I}$ be a common source multi-commodity instance, and 
Let $f$ be a feasible flow. Then $f$ is $\theta$-inducible if and only if $\hat{G}(f)$ does not contain a cycle of negative cost with respect to $c$. 
\label{thm:induc}
\end{theorem}

Theorem \ref{thm:induc} does not hold for general multi-commodity instances (see Remark \ref{rem:multi} in the appendix). 
%\begin{remark}
%The result of Theorem \ref{thm:induc} does not hold commodity-wise. That is, if we define $\hat{G}_i(f)$ for all $i \in [k]$ by replacing the set $\bar{A}$ in the definition of $\hat{G}(f)$ with $\bar{A}_i = \{(v,u) : a = (u,v) \in A^+_i\}$, then it is not true that $f$ is inducible if and only if all the graphs $\hat{G}_i(f)$ do not contain a negative cost cycle.
%\label{rem:separable}
%\end{remark}
The proof of Lemma~\ref{lem:cycles} follows directly from
%immediately from the non-negative cycle condition established in 
Theorem \ref{thm:induc}.
%We will now give some results that are used in order to upper bound the \ratio. Lemma \ref{lem:cycles} is a consequence of Theorem \ref{thm:induc}.

\begin{lemma}\label{lem:cycles}
Let $x$ be $\theta$-inducible and let $X_i$ be a flow-carrying $(s,t_i)$-path for commodity $i \in [k]$ in $G$. Let $\chi$ and $\psi$ be any $(s,t_i)$-path and $(t_i,s)$-path in $\hat{G}(x)$, respectively. Then 
\begin{align*}
\sum_{a \in X_i} l_a(x_a) + \theta_a^{\min}(x_a) & \leq \sum_{a \in \chi \cap A} l_a(x_a) + \theta_a^{\max}(x_a) - \sum_{a \in \chi \cap \bar{A}} l_a(x_a) + \theta_a^{\min}(x_a) \\
\sum_{a \in X_i} l_a(x_a) + \theta_a^{\max}(x_a) & \geq \sum_{a \in \psi \cap \bar{A}} l_a(x_a) + \theta_a^{\min}(x_a) - \sum_{a \in \psi \cap A} l_a(x_a) + \theta_a^{\max}(x_a).
\end{align*}
\end{lemma}

The following notion of alternating paths turns out to be crucial. It was first introduced by Lin et al. \cite{Lin2011} and is also used by Nikolova and Stier-Moses \cite{Nikolova2015}.
\begin{definition}[Alternating path \cite{Lin2011, Nikolova2015}]
Let $x$ and $z$ be feasible flows. We partition $A = X \cup Z$, where 
$Z = \{a \in A : z_a \geq x_a \text{ and } z_a > 0\}$ and $X = \{a \in A : z_a < x_a \text{ or } z_a = x_a = 0\}$. We say that $\pi_i = (a_1,\dots,a_r)$ is an alternating $s,t_i$-path if the arcs in $\pi_i \cap Z$ are oriented in the direction of $t_i$, and the arcs in $\pi_i \cap X$ are oriented in the direction of $s$.  
\label{def:alt_path}
\end{definition}

Without loss of generality we may remove all arcs with $z_a = x_a = 0$ (as they do not contribute to the social cost). Note that if along $\pi_i$ we reverse the arcs of $Z$ then the resulting path is a directed $(t_i,s)$-path in $\hat{G}(z)$ (which we call the \emph{$s$-oriented version of $\pi_i$}); similarly, if we reverse the arcs of $X$ then the resulting path is an $(s,t_i)$-path in $\hat{G}(x)$ (which we call the \emph{$t_i$-oriented version of $\pi_i$}).

The following lemma proves the existence of an \emph{alternating path tree}, i.e., a spanning tree of alternating paths, rooted at the common source node $s$.
It is a direct generalization of Lemma 4.6 in \cite{Lin2011} and Lemma 4.5 in \cite{Nikolova2015}. 
\begin{lemma}
%Let $\mathcal{I}$ be a common source multi-commodity instance.
Let $z$ and $x$ be feasible flows and let $Z$ and $X$ be a partition of $A$ as in Definition \ref{def:alt_path}. Then there exists an alternating path tree. 
\label{lem:alt_path_tree}
\end{lemma}

%\begin{lemma}
%Let $x$ be a \deviated{} Nash flow for some $\theta$-\deviation{} $\delta$, and let $p \in \mathcal{P}_i$ be a flow-carrying path for commodity $i \in [k]$, satisfying $\delta_p(x) = \min\{ \delta_q(x) : q \in \mathcal{P}_i \text{ and } x_q^i > 0 \}$.
%Then we have $\sum_{q \in \mathcal{P}_i} x_q^i l_q(x) \leq r_i l_p(x)$, where $x_i^q$ is the flow of commodity $i$ on $q$, i.e., $x_q = \sum_{i \in [k]} x_q^i$.\footnote{A similar result is given in the working paper \cite{pra_workdoc} of the work in \cite{Nikolova2015}.} \qed
%\label{lem:sc_bound}
%\end{lemma}
%Intuitively, Lemma \ref{lem:sc_bound} tells us that we can upper bound the contribution of commodity $i$, to the social cost, by the demand of commodity $i$ multiplied by the \textit{true} maximum latency for that commodity. Although all flow-carrying paths have the same \textit{\deviated{}} latency, this is not the case for the true latencies (measured in the social cost).

We now have all the ingredients to prove the following main result. 
%Theorem \ref{lem:main_bound} is a generalization of Theorem 4.6 \cite{Nikolova2015}.

%\gsrem{all main results of this section are stated in Thm. 3 now to avoid repetitions and be more concise}

\begin{theorem}
\label{lem:main_bound}
\label{thm:general_multi_common_source}
%Let $\mathcal{I}$ be a common source multi-commodity instance. 
Let $x$ be $\theta$-inducible and let $z$ be $0$-inducible. Further, let $A = X \cup Z$ be a partition as in Definition \ref{def:alt_path}. Let $\pi$ be an alternating path tree, where $\pi_i$ denotes the alternating $s, t_i$-path in $\pi$. 
\begin{enumerate}[(i)]
\item Suppose $\theta = (\theta^{\min}, \theta^{\max})$. Let $X_i$ be a flow-carrying path of commodity $i \in [k]$ maximizing $l_{P}(x)$ over all $P \in \mathcal{P}_i$.\footnote{Note that the values $l_P(x) + \delta_P(x)$ are the same for all flow-carrying paths, but this is not necessarily true for the values $l_P(x)$.} 
Then
$$
C(x) \leq C(z) + \sum_{i \in [k]} r_i \bigg( \sum_{a \in Z \cap \pi_i}\theta_a^{\max}(z_a) - \sum_{a \in X \cap \pi_i} \theta_a^{\min}(z_a) - \sum_{a \in X_i} \theta_a^{\min}(x_a) \bigg).
$$
\item Suppose $\theta = (\alpha, \beta)$ with $-1 < \alpha \leq 0 \leq \beta$. Let $\eta_i$ is the number of disjoint segments of consecutive arcs in $Z$ on the alternating $s, t_i$-path $\pi_i$ for $i \in [k]$.\footnote{Note that $\eta_i \le \lceil (n-1)/2 \rceil$.} 
Then
$$
\frac{C(x)}{C(z)}
\leq  1 + \frac{\beta - \alpha}{1 + \alpha}\cdot \sum_{i \in [k]} r_i\eta_i 
\leq  1 + \frac{\beta - \alpha}{1 + \alpha}\cdot \left \lceil \frac{n-1}{2}\right \rceil\cdot r.
$$
\end{enumerate}
\end{theorem}
\begin{proof}[i]
We have $C(x) = \sum_i \sum_{P \in \mathcal{P}_i} x_P^i l_P(x) \leq \sum_i r_i \sum_{a \in X_i} l_a(x_a)$ by the choice of $X_i$. By applying the first inequality of Lemma \ref{lem:cycles} to the flow $x$ in the graph $\hat{G}(x)$, where we choose $\chi$ to be the $t_i$-oriented version of $\pi_i$, we obtain
$$
\sum_{a \in X_i} l_a(x_a) + \theta_a^{\min}(x_a)  \leq   \sum_{a \in Z \cap \pi_i} l_a(x_a) + \theta_a^{\max}(x_a) - \sum_{a \in X \cap \pi_i} l_a(x_a) + \theta_a^{\min}(x_a).
$$

Let $Z_i$ be an arbitrary flow-carrying path of commodity $i \in [k]$ with respect to $z$. By applying the second inequality of Lemma \ref{lem:cycles} to the flow $z$ in the graph $\hat{G}(z)$ with $\theta^{\max} = \theta^{\min} = 0$, where we choose $\psi$ to be the $s$-oriented version of $\pi_i$, we obtain
$$
\sum_{a \in Z_i} l_a(z_a) \geq \sum_{a \in Z \cap \pi_i}l_a(z_a) - \sum_{a \in X \cap \pi_i}l_a(z_a).
%\footnote{The two obtained inequalities can be seen as generalizations of the results in Lemma 4.4 \cite{Nikolova2015} and Lemma 4.5 \cite{Nikolova2015} in the case that $k=1$.} 
$$
Combining these inequalities and exploiting the definition of $X$ and $Z$, we obtain 
\begin{align*}
\sum_{a \in X_i} l_a(x_a) + \theta_a^{\min}(x_a)
& \leq  \sum_{a \in Z \cap \pi_i} l_a(x_a) + \theta_a^{\max}(x_a) - \sum_{a \in X \cap \pi_i} l_a(x_a) + \theta_a^{\min}(x_a)  \nonumber \\
& \leq \sum_{a \in Z \cap \pi_i} l_a(z_a) + \theta_a^{\max}(z_a) - \sum_{a \in X \cap \pi_i} l_a(z_a) + \theta_a^{\min}(z_a) \nonumber \\
& \leq  \sum_{a \in Z_i} l_a(z_a) +  \sum_{a \in Z \cap \pi_i}\theta_a^{\max}(z_a) - \sum_{a \in X \cap \pi_i} \theta_a^{\min}(z_a).  \nonumber 
\end{align*}

%\gsrem{removed paragraph here}
%Note that, roughly speaking, $\chi$ and $\psi$ are each other's reversed path. In Lemma \ref{lem:cycles} the roles of $A$ and $\bar{A}$ then interchange. Furthermore, they then correspond to the sets $Z$ and $X$ in the alternating path setting (which is notation-wise easier to work with since technically $\chi$ and $\psi$ are part of different graphs).
The claim now follows by multiplying the above inequality with $r_i$ and summing over all commodities $i \in [k]$. Note that $C(z) = \sum_{i} r_i\sum_{a \in Z_i} l_a(z_a)$.
%, for any choice of flow-carrying paths $Z_i$ for commodities $i \in [k]$ in $z$, which is a well-known fact for Nash flows as in Wardop's model \cite{Wardrop1952}. 
\qed
\end{proof} 

%We obtain a more concise formulation of the bound in Theorem \ref{lem:main_bound} for $(\alpha,\beta)$-\deviations{}. 
\iffalse
\begin{theorem}
Let $-1 < \alpha \leq 0 \leq \beta$ be given constants. 
Let $x$ be $\delta$-inducible for some $(\alpha, \beta)$-\deviation{} $\delta$ and let $z$ be $0$-inducible. Let $\pi$ be an alternating path tree and let $\eta_i$ be the number of disjoint segments of consecutive arcs in $Z$ on the alternating $s, t_i$-path $\pi_i$ for $i \in [k]$.\footnote{Note that $\eta_i \le \lceil (n-1)/2 \rceil$ always.} Also, define $r = \sum_{i \in [k]} r_i$. Then
$$
\frac{C(x)}{C(z)} \leq \bigg( 1 + \frac{\beta - \alpha}{1 + \alpha}\cdot \sum_{i \in [k]} r_i\eta_i\bigg) \leq \bigg( 1 + \frac{\beta - \alpha}{1 + \alpha}\cdot \left \lceil \frac{n-1}{2}\right \rceil \cdot r\bigg).
$$
\label{thm:general_multi_common_source}
\end{theorem}
\fi
%\gsrem{removed par. here}
%The second inequality in the statement of Theorem \ref{thm:general_multi_common_source} follows from the fact that $\eta_i$ is upper bounded by $\lceil (n-1)/2 \rceil$. 
%\begin{corollary}
%Let $-1 < \alpha \leq 0 \leq \beta$ be given constants. Furthermore, let $\delta$ be an $(\alpha,\beta)$-\deviation{} that induces the \deviated{} Nash flow $x = f^\delta$, and let $z = f^0$ be a \classical{} Nash flow. Then, we have $C(x) \leq \left(1 + \frac{\beta - \alpha}{1 + \alpha}\cdot \left \lceil \frac{n-1}{2}\right \rceil \cdot \sum_{i \in [k]} r_i\right) C(z).$
%\label{cor:ratio}
%\end{corollary}

\section{Lower bounds  for $(\alpha,\beta)$-\deviations{}}

We show that the bound in Theorem \ref{thm:general_multi_common_source} is tight in all its parameters for $(\alpha, \beta)$-deviations. We start with single-commodity instances.

%The instances in Example \ref{exmp:tight_pra} show that the bound in Corollary \ref{cor:ratio} is tight for all $-1 < \alpha \leq 0 \leq \beta$ and $n \in \N$. Example \ref{exmp:tight_pra_multi_odd}, for the multi-commodity case with common source, shows that the bound in Corollary \ref{cor:ratio} is tight for all values of $r = \sum_{i \in [k]} r_i \in \R_{\geq 1}$ and $n = 2m+1 \in \N$. For $n = 2m \in \N$, we argue that, in the common source multi-commodity case, the bound of Corollary \ref{cor:ratio} can never be tight, but we give almost tight lower bounds (where the gap is independent of the network topology).

%\gsrem{put figure back in}

\begin{figure}[t!]
\centering
\scalebox{0.8}{
\begin{tikzpicture}[
  ->,
  >=stealth',
  shorten >=1pt,
  auto,
  %node distance=2cm,
  semithick,
  every state/.style={circle,fill=white,radius=0.2cm,text=black},
]
r = 1
\begin{scope}
  \node[state]  (s)               					 {$s$};
  \node[state]  (v3) [right=3cm of s] 				 {$v_3$};
  \node[state]  (v4) [above=0.75cm of v3]  			 {$v_4$};
  \node[state]  (v2) [below=0.75cm of v3]			     {$v_2$};
  \node[state]  (v1) [below right=0.75cm and 0.75cm of v2] {$v_1$};
  \node[state]  (w4) [above right=0.75cm and 0.75cm of v4] {$w_4$};
  \node[state]  (w3) [right=3cm of v4] 				 {$w_3$};
  \node[state]  (w2) [right=3cm of v3] 				 {$w_2$};       
  \node[state]  (w1) [right=3cm of v2] 				 {$w_1$}; 
  \node[state]  (t)  [right=3cm of w2] 				 {$t$};
  
\path[every node/.style={sloped,anchor=south,auto=false}]
(s) edge[line width=1.5pt, bend left=25] 	node {$(1,\beta)$} (w4)            
(s) edge[bend left=5]  						node {$(y_m(x),0)$} (v4)
(s) edge               						node {$(2y_m(x),0)$} (v3)
(s) edge[bend right=5] 						node {$(3y_m(x),0)$} (v2)
(s) edge[bend right=25]						node {$(4y_m(x),0)$} (v1)
(v4) edge[line width=1.5pt]             	node {$(1,0)$} (w4)
(v4) edge[line width=1.5pt]         		node {$(1,\beta)$} (w3)
(v3) edge[line width=1.5pt]               node {$(1,0)$} (w3)
(v3) edge[line width=1.5pt]               node {$(1,\beta)$} (w2)
(v2) edge[line width=1.5pt]               node {$(1,0)$} (w2)
(v2) edge[line width=1.5pt]               node {$(1,\beta)$} (w1)
(v1) edge[line width=1.5pt]               node {$(1,0)$} (w1)
(w4) edge[bend left=25]node {$(4y_m(x),0)$} (t)            
(w3) edge[bend left=5] node {$(3y_m(x),0)$} (t)
(w2) edge              node {$(2y_m(x),0)$} (t)
(w1) edge[bend right=5]node {$(y_m(x),0)$} (t)
(v1) edge[line width=1.5pt, bend right=25]node {$(1,\beta)$} (t);
        
\end{scope}
\end{tikzpicture}}
\caption{The fifth Braess graph with $(l_a^5,\delta_a^5)$ on the arcs as defined in Example \ref{exmp:tight_pra}. The bold arcs indicate the alternating path $\pi_1$.} 
\label{fig:braess_5}
\end{figure}
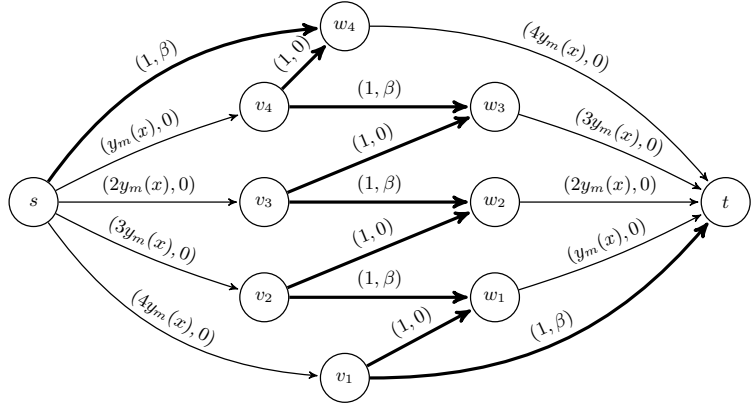

%\gsrem{made definition inline (not referenced a lot)}

%\begin{definition}[$m$-th Braess graph \cite{Roughgarden2006}]
Our instance is based on the generalized Braess graph \cite{Roughgarden2006}. 
The \emph{$m$-th Braess graph} $G^m = (V^m, A^m)$ is defined by $V^m = \{s,v_1,\dots,v_{m-1},w_1,\dots,w_{m-1},t\}$ and $A^m$ as the union of three sets: $E^m_1 = \{(s,v_j),(v_j,w_j),(w_j,t): 1 \leq j \leq m-1\}$, $E^m_2 = \{(v_j,w_{j-1}) : 2 \leq j \leq m\}$ and $E^m_3 = \{(v_1,t)\cup\{(s,w_{m-1}\}\}$. 
%\end{definition}
%\noindent Note that $B^1$ is the well-known Braess graph consisting of four nodes. The tightness example for $n = 2$ is left to the reader.

%\gsrem{Guido: I think we claim tightness for single-commodity for all $n$; example seems to work for $n$ even only..? }

%\gsrem{Pieter: The value $\lceil[(n-1)/2]\rceil$ is the same for two consecutive values $n = 2m, n+1 = 2m+1$. E.g., for $n = 4,5$ we have $\lceil[(n-1)/2]\rceil = 2$. We already show tightness for $n = 4$, so $n = 5$ follows trivially by adding a redundant node to the instance. ADDED COMMENT AFTER EXAMPLE 1}
%\pkrem{good. was my personal confusion; a remark in a  footnote will do; changed it}

\begin{example}%[Single-commodity lower bound]
By Lemma \ref{lem:alphabeta_induc} (see appendix), we can assume without loss of generality that $\alpha = 0$. Let $\beta \geq 0$ be a fixed constant and let $n = 2m \geq 4\in \N$.\footnote{Note that the value $\lceil(n-1)/2\rceil$ is the same for $n  \in \{2m,2m+1\}$ with $m \in \N$. The example shows tightness for $n = 2m$. The tightness for $n = 2m+1$ then follows trivially by adding a dummy node.}
Let $G^{m}$ be the $m$-th Braess graph. Furthermore, let $y_m : \R_{\geq 0} \rightarrow \R_{\geq 0}$ be a non-decreasing, continuous function\footnote{%
%The latency functions used in (Proposition 4.4, \cite{Roughgarden2006}), and those for the construction in \cite{Lianeas2015}, are closely related to the functions given here. 
For example $y_m(g) = m(m-1)\beta\max\{0,\left(g - \frac{1}{m}\right)\}$. That is, we define $y_m$ to be zero for $0 \leq g \leq 1/m$ and we let it increase with constant rate to $\beta$ in $1/(m-1)$.} with $y_m(1/m) = 0$ and $y_m(1/(m-1)) = \beta$.
We define 
$$
l^{m}_a(g) = \left\{ \begin{array}{ll}
(m - j)\cdot y_m(g) & \text{ for } a \in \{(s,v_j) : 1 \leq j \leq m-1\}\\
j \cdot y_m(g) & \text{ for } a \in \{(w_j,t) : 1 \leq j \leq m-1\}\\
1 & \text{ otherwise.}
\end{array}\right.
$$
Furthermore, we define $\delta^{m}_a(g) = \beta$ for $a \in E^m_2$, and $\delta^{m}_a(g) = 0$ otherwise.
%$$
%\delta^{m}_a(g) = \left\{ \begin{array}{ll}
%\beta  & \text{ for } a \in E^m_2\\
%0 & \text{ otherwise.}
%\end{array}\right.
%$$
Note that $0 \leq \delta_a^{m}(g) \leq \beta l_a^{m}(g)$ for all $a \in A$ and $g \geq 0$ (see Figure \ref{fig:braess_5} in the appendix).
%\pkrem{Figure moved to Appendix B}

A Nash flow  $z = f^0$ is given by routing $1/m$ units of flow over the paths $(s,w_{m-1},t),(s,v_1,t)$ and the paths in $\{(s,v_j,w_{j-1},t) : 2 \leq j \leq m-1\}$. Note that all these paths have latency one, and the path $(s,v_j,w_j,t)$, for some $1 \leq m \leq j$, also has latency one. We conclude that $C(z) = 1$.

A Nash flow $x = f^\delta$, with $\delta$ as defined above, is given by routing $1/(m-1)$ units of flow over the paths in $\{(s,v_j,w_j,t) : 1 \leq j \leq m - 1\}$. Each such path $P$ then has a latency of $l_P(x) = 1 + \beta m$. It follows that $C(x) = 1 + \beta m$. Note that the deviated latency of path $P$ is $q_P(x) = 1 + \beta m$ because all deviations along this path are zero. Each path $P' = (s,v_j,w_{j-1},t)$, for $2 \leq j \leq m-1$, has a \deviated{} latency of $q_{P'}(x) = 1 + \beta + (m-1) y_m(1/(m-1)) = 1 + \beta + (m - 1)\beta = 1 + \beta m$. 
The same argument holds for the paths $(s,w_{m-1},t)$ and $(s,v_1,t)$. We conclude that $x$ is $\delta$-inducible. It follows that $C(x)/C(z) = 1 + \beta m = 1 + \beta n /2$. \qed
%\indent Combining the results, we find
%$$
%\frac{C(x)}{C(z)} = \frac{1 + \frac{\beta - \alpha}{1 + \alpha}m}{1} = 1 + \left(\frac{\beta - \alpha}{1 + \alpha}\right)\frac{n}{2}.
%$$ \qed
\label{exmp:tight_pra}
\end{example} 

By adapting the construction above, we obtain the following result. 
% \gsrem{Comment on odd $n$}

%\gsrem{statement in terms of deviation ratio now}

\begin{theorem}\label{thm:lower_multi_odd}
There exist common source two-commodity instances $\mathcal{I}$ such that 
$$
\dr(\mathcal{I}, (\alpha, \beta)) 
\ge 
\left\{ \begin{array}{ll}
1 + (\beta - \alpha)/(1+\alpha)\cdot (n-1)/2 \cdot r & \text{ for } n = 2m+1 \in \N_{\geq 5} \\
1 + (\beta - \alpha)/(1+\alpha)\cdot [(n/2-1)r + 1] & \text{ for } n = 2m \in \N_{\geq 4}.
\end{array}\right.
$$
%where all the parameters are as in Theorem \ref{thm:general_multi_common_source}.
\end{theorem}

%\gsrem{moved further details to the appendix}

For two-commodity instances and $n$ even, we can actually improve the upper bound in Theorem \ref{thm:general_multi_common_source} to the lower bound stated in Theorem \ref{thm:lower_multi_odd} (see Remark~\ref{rem:improve-2-com} in the appendix). 

For general multi-commodity instances the situation is much worse. In particular, we establish an exponential lower bound on the Deviation Ratio. The instance used in proof of Theorem~\ref{thm:multi_general} is similar to the one used by Lin et al. \cite{Lin2011}.

%The upper bound in Theorem \ref{thm:general_multi_common_source} no longer holds for general multi-commodity instances.

%\gsrem{statement in terms of deviation ratio now; removed the mentioning of the constant $c$ and put it online to reduce number of parameters}

\begin{theorem}
For every $p = 2q + 1 \in \N$, there exists a two-commodity instance $\mathcal{I}$ whose size is polynomially bounded in $p$ such that
%$|V|,|A| \in \mathcal{O}(p)$, for which 
$
\dr(\mathcal{I}, (\alpha, \beta))
\geq 1 + \beta F_{p+1} \approx 1 + 0.45 \beta \cdot \phi^{p+1},
$
where $F_p$ is the $p$-th Fibonacci number 
%$c\approx 0.4472$ 
and $\phi \approx 1.618$ is the golden ratio.
%\footnote{Remember that $F_p = F_{p-2} + F_{p-1}$ for $p \geq 2$, where $F_0 = 0$ and $F_1 = 1$.}
\label{thm:multi_general}
\end{theorem}

\section{Applications}

%\gsrem{moved application back in}

By using our bounds on the \ratio{}, we obtain the following results. 

\paragraph{\textbf{Price of Risk Aversion}.} 

%\gsrem{statement in terms of PRA now; precise lower bound statements deferred to appendix; removed some details and made it part of the proof (given in appendix)}

\begin{theorem}
The Price of Risk Aversion for a common source multi-commodity instance $\mathcal{I}$ with non-negative and non-decreasing latency functions, variance-to-mean-ratio $\kappa > 0$ and risk-aversion parameter $\gamma \ge -1/\kappa$ is at most 
$$
\text{PRA}(\mathcal{I}, \gamma, \kappa) \le 
\begin{cases} 
1 - \gamma \kappa /(1 + \gamma \kappa ) \lceil (n-1)/2 \rceil r & \text{for $-1/\kappa < \gamma \leq 0$} \\
1 + \gamma \kappa \lceil (n-1)/2 \rceil r & \text{for $\gamma \ge 0$}.
\end{cases}
$$
Moreover, these bounds are tight in all its parameters if $n = 2m+1$ and almost tight if $n = 2m$ (see appendix for precise statements). 
In particular, for single-commodity instances we obtain tightness for all $n \in \N$.
\label{thm:pra_model_results}
\end{theorem}

\paragraph{\textbf{Stability of Nash flows under small perturbations.}}  

%\gsrem{result put in theorem format for consistency and reference}

\begin{theorem}
\label{thm:stability}
Let $\mathcal{I}$ be a common source multi-commodity instance with non-negative and non-decreasing latency functions $(l_a)_{a \in A}$. Let $f$ be a Nash flow with respect to $(l_a)_{a \in A}$ and let $\tilde{f}$ be a Nash flow with respect to slightly perturbed latency functions $(\tilde{l}_a)_{a \in A}$ satisfying 
$
\sup_{a \in A,\; x \geq 0} | (l_a(x) - \tilde{l}_a(x))/l_a(x) | \leq \epsilon
$
for some small $\epsilon > 0$.
Then the relative error in social cost is $(C(\tilde{f}) - C(f))/C(f) \le 
2\epsilon/(1 - \epsilon) \lceil(n-1)/2 \rceil \cdot r 
= \mathcal{O}(\epsilon rn)$.
\end{theorem}

\section{Smoothness based approaches}

We derive tight smoothness bounds on the Biased Price of Anarchy for $(0, \beta)$-deviations. Our bounds improve upon the bounds of $(1 + \beta)/(1 - \mu)$ recently obtained by Meir and Parkes \cite{Meir2015} and Lineas et al. \cite{Lianeas2015} for $(1, \mu)$-smooth latency functions. As a direct consequence, we also obtain better smoothness bounds on the Price of Risk Aversion.
% (but no tightness). 
Our approach is a generalization of the framework of Correa, Schulz and Stier-Moses \cite{Correa2008} (which we obtain for $\beta = 0$). 

Let $\mathcal{L}$ be a given set of latency functions and $\beta \geq 0$ fixed. For $l \in \mathcal{L}$, define
$$
\hat{\mu}(l,\beta) = \sup_{x,z \geq 0} \left\{ \frac{z[l(x) - (1+\beta)l(z)]}{xl(x)} \right\}
\quad\text{and}\quad
\hat{\mu}(\mathcal{L},\beta) = \sup_{l \in \mathcal{L}} \hat{\mu}(\mathcal{L},\beta).
$$ 

\begin{theorem}\label{thm:smooth-BPA}
Let $\mathcal{L}$ be a set of non-negative, non-decreasing and continuous functions. Let $\mathcal{I}$ be a general multi-commodity instance with $(l_a)_{a \in A} \in \mathcal{L}^A$. Let $x$ be $\delta$-inducible for some $(0,\beta)$-\deviation{} $\delta$ and let $z$ be an arbitrary feasible flow.
Then $C(x)/C(z) \leq (1 + \beta)/(1 - \hat{\mu}(\mathcal{L},\beta))$ if $\hat{\mu}(\mathcal{L},\beta) < 1$. Moreover, this bound is tight if $\mathcal{L}$ contains all constant functions and is closed under scalar multiplication, i.e., for every $l \in \mathcal{L}$ and $\gamma \geq 0$, $\gamma l \in \mathcal{L}$.
\label{thm:smoothness}
\end{theorem}

%\gsrem{moved affine latency example back in to give some intuition of the result; also put back in the gap result; mentioned here with pointer to appendix}

For example, for affine latencies $\hat{\mu}(\mathcal{L},\beta) = 1/(4(1+\beta))$ (see Proposition~\ref{prop:mu-affine} in the appendix) and we obtain a bound of $(1+\beta)^2/(\frac34 + \beta)$ on the Biased Price of Anarchy, which is strictly better than the bound $4(1 + \beta)/3$ obtained in \cite{Lianeas2015,Meir2015}. 

We also provide an upper bound on the absolute gap between the Biased Price of Anarchy and the \ratio{} (see Corollary~\ref{cor:smoothness_approx} in the appendix). 

\bigskip
\noindent
As a final result we derive smoothness bounds for general path deviations, which are not necessarily decomposable into arc deviations. The main motivation for investigating such deviations is that we can apply such bounds to the \textit{mean-std} objective of the Price of Risk Aversion model by Nikolova and Stier-Moses \cite{Nikolova2015} (see Section \ref{sec:pre}).
We need to adjust some definitions of Section \ref{sec:pre}. We are given non-positive and non-negative, respectively, continuous threshold functions $\theta^{\min} = \left(\theta^{\min}_P\right)_{P \in \mathcal{P}}$ and $\theta^{\max} = \left(\theta^{\max}_P\right)_{P \in \mathcal{P}}$ and consider \deviations{} 
$(\delta_P)_{P \in \mathcal{P}}$ 
from 
$$
\Delta(\theta) = \{(\delta_P)_{P \in \mathcal{P}} \ : \text{$\theta^{\min}_P(f) \leq \delta_P(f) \leq \theta^{\max} _P(f)$ for all feasible flows $f$}\}.
$$ 
Now $(\alpha,\beta)$-\deviations{} are \deviations{} $\delta \in \Delta(\theta)$ with $\theta^{\min}_P = \alpha l_P$ and $\theta^{\max}_P = \beta l_P$ for all $P \in \mathcal{P}$. 

Let $f$ be $\delta$-inducible with respect to some $(\alpha, \beta)$-deviation $\delta$. The Nash flow conditions (\ref{eq:nash}) then imply that $\forall i \in [k], \forall P \in \mathcal{P}_i, f_P > 0$:
$$
(1 + \alpha)l_P(f) \leq  l_P(f) + \delta_P(f) \leq l_{P'}(f) + \delta_{P'}(f) \leq (1+\beta)l_{P'}(f)  \ \ \forall P' \in \mathcal{P}_i.
$$
In particular, the above inequality reveals that $f$ is an \textit{$(1+\beta)/(1+\alpha)$-approximate Nash flow} (see \cite{Christodoulou2011}). As a consequence, the bounds by Christodoulou et al. \cite{Christodoulou2011}, on the Price of Anarchy for approximate Nash flows in non-atomic routing games with polynomial latency functions, yield upper bounds on the BPoA and DR of instances with polynomial latency functions.

%Also, in \cite{Nikolova2015} and \cite{Lianeas2015}, some topological bounds are obtained for the non-linear \textit{mean-std} objective in single-commodity instances. It is shown that the upper bounds in Theorem \ref{thm:general_multi_common_source} also hold for $\eta_1 \in \{1,2\}$ (and non-negative $\gamma$) for the \textit{mean-std} objective.
%We will now, to the best of our knowledge, introduce the first smoothness bounds that apply to \textit{non-linear} \deviations{}, that is,  the \deviation{} along a path does not necessarily have to be induced as a sum of arc \deviations{}. 

%In terms of smoothness, we obtain the bound stated in the theorem below. 

%For atomic congestion games, similar results are known (see, e.g., Roughgarden \cite{Roughgarden2015}).

\begin{theorem}
Let $\mathcal{I}$ be a general multi-commodity instance with $(l_a)_{a \in A} \in \mathcal{L}^A$. Let $x$ be $\delta$-inducible with respect to some $(0,\beta)$-path \deviation{} $\delta$ and let $z$ an arbitrary feasible flow.
If $\hat{\mu}(\mathcal{L},0) < 1/(1 + \beta)$, then $C(x)/C(z) \leq (1 + \beta)/(1 - (1 + \beta)\hat{\mu}(\mathcal{L},0))$.
\label{thm:smooth_general}
\end{theorem}

%We conclude this section with some comments on the nature of smoothness bounds. Although they provide insightful information, in particular giving mayor improvements on the topological bounds in some cases, they do not capture the gap between the \classical{} Nash flow and the \deviated{} Nash flow, but rather the gap between the social optimum and the \deviated{} Nash flow, which Meir and Parkes \cite{Meir2015} call the \textit{Biased Price of Anarchy}. This is true since the obtained bound in Theorem \ref{thm:smoothness} holds for every feasible flow. Moreover, we actually know what the  Price of Risk Aversion should be for $\gamma \kappa = 0$, namely $1$, whereas the smoothness bounds give the Price of Anarchy (see Koutsoupias et al. \cite{Koutsoupias1999}) for $\gamma\kappa = 0$. This also means that, for a fixed instance with Price of Anarchy strictly greater than $1$, any continuous upper bound $U(\gamma\kappa)$ on the Price of Risk Aversion, satisfying $U(0) = 1$, will give a better upper bound for $\gamma\kappa> 0$ small enough (which is the case for the topological upper bounds). That is, for small values of $\gamma \kappa$, the obtained smoothness bounds are most likely not tight for the \ratio{}.\bigskip

%\section{A preliminary result for heterogeneous players} 

\section{Conclusions}

%\gsrem{added this section}

We introduced a unifying model to study the impact of (bounded) worst-case latency deviations in non-atomic selfish routing games. We demonstrated that the Deviation Ratio is a useful measure to assess the cost deterioration caused by such deviations. Among potentially other applications, we showed that the Deviation Ratio provides bounds on the Price of Risk Aversion and the relative error in social cost if the latency functions are subject to small perturbations. 

Our approach to bound the Deviation Ratio (see Section~\ref{sec:upper-bounds}) is quite generic and, albeit considering a rather general setting, enables us to obtain tight bounds. We believe that this approach will turn out to be useful to derive bounds on the Deviation Ratio of other games (e.g., network cost sharing games).

A natural extension of the bounded deviation model introduced in Section \ref{sec:pre} is to consider \emph{heterogeneous players}, i.e., players have different attitudes towards the deviations. Below we briefly report on some preliminary results for single-commodity networks. These extensions also hold for the framework of path deviations as described in the previous section. 

In general, studying the impact of (bounded) worst-case deviations of the input data of more general classes of games (e.g., congestion games) is an interesting and challenging direction for future work.

\paragraph{Preliminary results for single-commodity networks and heterogenous players.}

We consider $k$ different \emph{player types} in a single-commodity network (i.e., all player types share the same source and destination). For each type $i \in [k]$ we have a demand $r_i$ and an attitude $\tau_i$ towards the deviations. We assume without loss of generality that the demands are normalized such that $\sum_{i \in [k]} r_i = 1$. A feasible flow $f = (f_P^i)_{i \in [k], P \in \mathcal{P}}$ is \emph{$\delta$-inducible} if:
$$
\forall i \in [k],\ \forall P \in \mathcal{P},\ f_P^i > 0: \quad l_P(f) + \tau_i\delta_P(f) \leq l_{P'}(f) + \tau_i\delta_{P'}(f)  \quad \forall P' \in \mathcal{P}.
$$

We prove the following result: 
\begin{lemma}\label{lem:upper-heterogeneous}
Let $\mathcal{I}$ be a single-commodity instance and let $z$ be a $0$-inducible Nash flow. Let $x$ be a $\delta$-inducible Nash flow for some $(0,\beta)$-path deviation $\delta$. If there is an alternating $(s,t)$-path $\pi$ consisting only of arcs in $Z$, then 
$$
\frac{C(x)}{C(z)} \leq 1 + \beta \bigg(\sum_{i \in [k]} \tau_i r_i\bigg).
$$
\end{lemma}

Note that the condition of the alternating path $\pi$ to consist of arcs in $Z$ only is equivalent to having $\eta = 1$, i.e., $\pi$ is an actual $(s,t)$-path in the underlying graph. In particular, this condition is satisfied for \emph{series-parallel} graphs (see, e.g., Corollary 4.8 \cite{Nikolova2015}). This implies that the bound derived above holds for all instances with series-parallel graphs. It would be interesting to see if this bound extends to arbitrary alternating paths.

\begin{proof}[Lemma~\ref{lem:upper-heterogeneous}]
For $i \in [k]$, let $\bar{P}_i$ be a path maximizing $l_P(x)$ over all flow-carrying paths $P \in \mathcal{P}$ of type $i$. We have (this argument is also used in the proof of Lemma 4 in \cite{Lianeas2015}):
$$
l_{\bar{P}_i}(x) \leq l_{\bar{P}_i}(x) + \tau_i \delta_{\bar{P}_i}(x) \leq l_{\pi}(x) + \tau_i \delta_{\pi}(x) \leq  (1 + \beta \tau_i)l_{\pi}(x) = (1 + \beta \tau_i)\sum_{a \in \pi} l_a(x_a).
$$
Note that, by definition of the alternating path $\pi$, we have $x_a \leq z_a$ for all $a \in \pi$. Continuing with the estimate, we find $l_{\bar{P}_i}(x) \leq (1 + \beta \tau_i)\sum_{a \in \pi} l_a(z_a)$ and thus
%\begin{eqnarray}
%C(x) &=& \sum_{i=1}^i \sum_{P \in \mathcal{P}} x_P^i l_P(x) \leq \sum_{i=1}^i r^i l_{\bar{P}_i}(x) \nonumber \\
%& \leq & \sum_{i=1}^i r^i (1 + \beta \tau_i)\sum_{a \in \pi} l_a(z_a) \nonumber \\
%&=& C(z) \left( \sum_{i=1}^i r^i (1 + \beta \tau_i)\right) \nonumber
%\end{eqnarray}
%= \sum_{i=1}^i \sum_{P \in \mathcal{P}} x_P^i l_P(x)
$$
C(x) \leq \sum_{i \in [k]} r_i l_{\bar{P}_i}(x) \leq  \sum_{i \in [k]} r_i (1 + \beta \tau_i)\sum_{a \in \pi} l_a(z_a) = C(z) \bigg(\sum_{i \in [k]} r_i (1 + \beta \tau_i)\bigg) 
$$
Since $\sum_{i \in [k]} r_i = 1$, we get the desired result. Note that we use $C(z) = \sum_{a \in \pi} l_a(z_a)$, which is true because there exists a flow-decomposition of $z$ in which $\pi$ is flow-carrying (here we use $z_a > 0$ for all $a \in \pi$). \qed
\end{proof}

% \noindent \textbf{Acknowledgements.} We thank Crist\'obal Guzm\'an for the reference to \cite{Cominetti2015}.
\bibliographystyle{abbrv}
\bibliography{references}
\newpage

\appendix

%\gsrem{messed around in the appendix quite a bit; please check; most importantly: all statements of theorems, lemmas, etc. are repeated here to make it self-contained and facilitate reading this material; also subsections are introduced to facilitate finding the respective result/proof.} 

\section{Omitted material of Section \ref{sec:pre}}

\begin{rtheorem}{Theorem}{\ref{thm:np-hard_max}}
%It is \np-hard to compute deviations $\delta \in \Delta(\theta)$ such that $C(f^\delta)$ is maximized, even for single-commodity networks with linear latencies.
Given an instance $\mathcal{I}$, threshold functions $\theta$ and a parameter $K$, it is \np-complete to determine whether there exist deviations $\delta \in \Delta(\theta)$ such that $C(f^\delta) \ge K$, even for single-commodity networks with linear latencies.
\end{rtheorem}
\begin{proof}%[Theorem \ref{thm:np-hard_max}]
%\gsrem{Pieter: Changed it to DHP with fixed source/destination}
%\pkrem{Good! Perhaps even better rename it to \textsc{Directed Hamiltonian $s,t$-Path}. Also: would be good to mention that the DHstP problem is NP-complete.}
We give a reduction from the \textsc{Directed Hamiltonian $s,t$-Path} problem: We are given a directed graph $G = (V,A)$, and fixed $s,t \in V$, and the goal is to decide whether or not there exists a simple directed $s,t$-path in $G$ that visits every node exactly once. Let $\mathcal{J}$ be an instance of \textsc{Directed Hamiltonian $s,t$-Path} problem. 

 Now, define an instance $\mathcal{I}$ of the bounded deviation model on the graph $G$ by taking $l_a(x) = x$ for all $a \in A$, $\theta^{\min}_a = 0$ for all $a \in A$, and $\theta^{\max}_a = n - 1$ for all $a \in A$. Furthermore, take $r = 1$.\\
\indent We claim that $G$ has a Hamiltonian path from $s$ to $t$ if and only if there is a \deviation{} $\delta \in \Delta(\theta)$ such that $C(f^\delta) \geq n-1$. First, let $G$ have a Hamiltonian path $P$ from $s$ to $t$, and define $\delta$ by $\delta_a = 0$ if $a \in P$, and $\delta_a = n - 1$ otherwise.
%$$
%\delta_a = \left\{ \begin{array}{ll} 0 & \text{if } a \in P\\ 
%n - 1  & \text{otherwise}
%\end{array}  \right.
%$$
We then have that $f^\delta$ is given by $f_a^\delta = 1$ if $a \in P$ and $f_a = 0$ otherwise,
%$$
%f^\delta_a = \left\{ \begin{array}{ll} 1 & \text{if } a \in P\\ 
%0  & \text{otherwise}
%\end{array}  \right.
%$$
since the perceived latency along $P$ is then equal to $l_P(f) = n-1$, and any other path $P'$ uses at least one different arc $a' \notin P$, which gives us that $Q_{P'}(f) \geq l_{a'}(f) + \delta_{a'}(f) \geq n-1 = Q_{P}(f)$. Note that $f^\delta$ is the unique Nash flow in this case (since all the perceived latencies $l_a + \delta_a$ are strictly increasing).\\
\indent Conversely, suppose there is a $\delta \in \Delta(\theta)$ such that $C(f^\delta) \geq n - 1$. For any \textit{feasible} flow $g$ we have that $l_P(g) \leq n - 1$, with strict inequality if $f_P < 1$ (since then there will be at least one arc $a \in P$ with $f_a < 1$). This means that 
$$C(g) = \sum_{P \in \mathcal{P}} g_Pl_P(g) \leq \sum_{P \in \mathcal{P}} g_P(n-1) = n - 1,$$ 
using that $r = 1$. Again, we have strict inequality if $0 < g_P < 1$ for some path $P$, i.e., if not all players use the same path. This means that for $f^\delta$ there is at most one path $P^*$ with $f^\delta_{P^*} > 0$, which then implies that $f_{P^*} = 1$. Furthermore, we can conclude that $|A(P^*)| = l_{P^*}(f^\delta) = C(f^\delta) =  n - 1$, which implies that $P^*$ is a Hamiltonian path from $s$ to $t$, since it is a simple path by assumption. \qed
\end{proof}

\newpage

\section{Omitted material of Section \ref{sec:upper-bounds}}

\subsection{Proof of Theorem \ref{thm:induc}} 
\begin{rtheorem}{Theorem}{\ref{thm:induc}}
Let $f$ be a feasible flow. Then $f$ is $\theta$-inducible if and only if $\hat{G}(f)$ does not contain a cycle of negative cost with respect to $c$. 
\end{rtheorem}

%\gsrem{might be worth to go over this proof again and make it really crisp; this is one of the first (and perhaps only) proofs considered by the reviewers}

\begin{proof}%[Theorem \ref{thm:induc}]
Suppose that $f$ is an inducible flow and let $\delta$ be a vector of \deviations{} that induce $f$. Let $\hat{B}$ be a directed cycle in $\hat{G}(f)$. 
If $\hat{B}$ only consists of forward arcs, then $\sum_{a \in \hat{B}} (l_a + \theta^{\max}_a) \geq \sum_{a \in \hat{B}} (l_a + \theta^{\min}_a) \geq 0$, where the last inequality holds because of Assumption \ref{ass:nonnegative}. 

Next, suppose that there is a reversed arc $a = (v,u) \in \hat{B} \cap \bar{A}$. Then $(u,v) \in A_i^+$ for some commodity $i \in [k]$. Let $B = (b_1,\dots,b_q,b_1)$ be the cycle that we obtain from $\hat{B}$ if all arcs $(v,u) \in \hat{B} \cap \bar{A}$ are replaced by $a = (u,v) \in A^+$ (note that $B$ is contained in $G$ and that it is not a directed cycle). For every arc $b = (b_l,b_{l+1}) \in B \cap A^+$, there is a flow-carrying path $P_l$\footnote{Note that the paths $P_l$ can overlap, use parts of $B$, or even be subpaths of each other.} from $s$ to $b_l$ for some commodity $i$ (here we use the fact that all commodities share the same source).  

Intuitively, the proof is as follows. For all nodes $b \in V(B)$ with two incoming arcs of $B$, we can can find two paths $Q_1$ and $Q_2$ leading to that node, using the paths $P_l$ and the cycle $B$ (see also Figure \ref{fig:char_proof}). Furthermore, one of those paths is flow-carrying by construction. We then apply the Nash conditions to those flow-carrying paths (exploiting the common source) and add up the resulting inequalities. The contributions of the paths $P_l$ cancel out in the aggregated inequality, leading to the desired result. We now give a formal proof of this sketch.
%\gsrem{Added paragraph with sketch.}
%\pkrem{Very good. Could this actually become a proof sketch? Would be still nice to add a sketch of this theorem to the main text I think; seems we still have a bit of space.. ACTUALLY NOT ANY MORE.. NEVER MIND!}
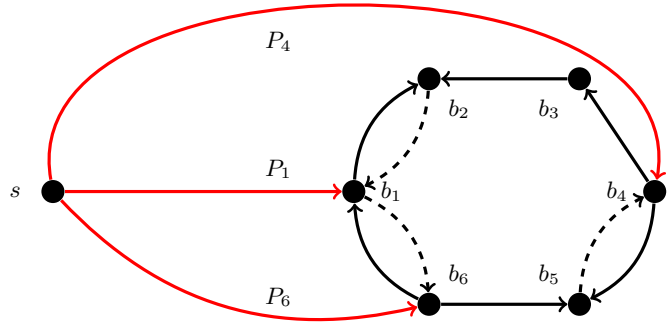
\begin{figure}[h!]
\centering
\begin{tikzpicture}
%\useasboundingbox (-8,-2) rectangle (5.5,3);   
\begin{scope}%[xshift=-3cm,yshift=-1cm] 

%Examples
%\node at (-1,2) [circle,fill=black] (s) {}; 
%\node at (0,2) [circle,fill=black] (1) {}; 
%\draw (s) edge[->, very thick, red, bend left] (1);

\node at (-1,1.5) [circle,fill=black] (2) {}; 
\node at (1,1.5)  [circle,fill=black] (3) {}; 
\node at (-2,0) [circle,fill=black] (1) {}; 
\node at (2,0)  [circle,fill=black] (4) {}; 
\node at (-1,-1.5)[circle,fill=black] (6) {}; 
\node at (1,-1.5) [circle,fill=black] (5) {}; 
\node at (-6,0) [circle,fill=black] (s) {}; 

\node at (-6.5,0) [] {$s$};
\node at (-1.5,0) [] {$b_1$};
\node at (-0.6,1.1) [] {$b_2$};
\node at (0.6,1.1) [] {$b_3$};
\node at (1.5,0) [] {$b_4$};
\node at (0.6,-1.1) [] {$b_5$};
\node at (-0.6,-1.1) [] {$b_6$};

\node at (-3,2) [] {$P_4$};
\node at (-3,0.3) [] {$P_1$};
\node at (-3,-1.4) [] {$P_6$};

%\draw (s) edge[->, very thick, red, bend left] (1);
\draw (1) edge[->, very thick, bend left] (2);
\draw (2) edge[->, very thick, bend left, dashed] (1);
\draw (3) edge[->, very thick] (2);
\draw (4) edge[->, very thick] (3);
\draw (4) edge[->, very thick, bend left] (5);
\draw (5) edge[->, very thick, bend left, dashed] (4);
\draw (6) edge[->, very thick] (5);
\draw (6) edge[->, very thick, bend left] (1);
\draw (1) edge[->, very thick, bend left, dashed] (6);

\draw (s) edge[->, very thick, red] (1);
\draw (s) edge[->, very thick, bend right, red] (6);
\draw (s) edge[->, very thick, bend left=100, red] (4);

\end{scope}
\end{tikzpicture}
\caption{The dashed arcs are the reversed arcs in $\hat{G}$. The black bold arcs indicate the cycle $B$. We have $(h_0,h_1,h_2,h_3) = (1,4,6,1)$. Note that, for example, it could be the case that $P_1 = P_6 \cup (b_6,b_1)$.}
\label{fig:char_proof}
\end{figure}

Without loss of generality, we may assume that $(b_1,b_2) \in A^+$. 
%\gsrem{Guido: mean $A^+$ here? Pieter: Yes, also changed the $\bar{A}$ in the next line to $A^+$}
Let $h_1 \in \{2,\dots,q+1\}$ be the smallest index for which $(b_{h_1},b_{h_1+1}) \in A^+$ (here we take $b_{q+1} := b_1$ and $P_{q+1} := P_1$). Note that the concatenation of $P_{h_1}$ and $(b_{h_1},b_{h_1-1},\dots,b_2)$ is a directed path from $s$ to $b_2$. Then we have
$$
l_{(b_1,b_2)} + \delta_{(b_1,b_2)} + \sum_{a \in P_1} (l_a + \delta_a) \leq  \sum_{j = 3}^{h_1} l_{(b_j,b_{j-1})} + \delta_{(b_j,b_{j-1})} +  \sum_{a \in P_{h_1}} (l_a + \delta_a) 
$$
by using the fact that a subpath $(s,\dots,u)$ of a shortest $(s,t_i)$-path $(s,\dots,u,\dots,t_i)$ is a shortest $(s,u)$-path if $G$ does not contain negative cost cycles under the cost function $l + \delta$ (which is true because of Property \ref{ass:nonnegative}). We can now repeat this procedure by letting $h_2 \in \{h_1 + 1,\dots, q+1\}$ be the smallest index for which $(b_{h_2},b_{h_2+1}) \in A^+$, then we have
$$
l_{(b_{h_1},b_{h_1 + 1})} + \delta_{(b_{h_1},b_{h_1 + 1})} + \sum_{a \in P_{h_1}} (l_a + \delta_a) \leq  \sum_{j = h_1 + 2}^{h_2} l_{(b_j,b_{j-1})} + \delta_{(b_j,b_{j-1})} +  \sum_{a \in P_{h_2}} (l_a + \delta_a). 
$$
Continuing this procedure, we find a sequence $1 = h_0 < h_1 < \dots < h_p = q+1$ such that, for every $0 \leq w \leq p-1$,
\begin{equation}
l_{(b_{h_w},b_{h_w + 1})} + \delta_{(b_{h_w},b_{h_w + 1})} + \sum_{a \in P_{h_w}} l_a + \delta_a \leq  \sum_{j = h_w + 2}^{h_{w+1}} l_{(b_j,b_{j-1})} + \delta_{(b_j,b_{j-1})} +  \sum_{a \in P_{h_{w+1}}} l_a + \delta_a .
\label{eq:char_proof}
\end{equation}
Note that $p$ is the number of reversed arcs on the cycle $\hat{B}$.

%Note that for a fixed flow $f$ the values $l_a(f_a)$, $\delta_a(f_a)$, $\theta^{\min}_a(f_a)$ and $\theta^{\max}_a(f_a)$ are constants. For the sake of conciseness, we  omit the explicit reference to the flow value $f_a$ below. 
%Suppose that $f$ is $\theta$-inducible and let $\delta \in \Delta(\theta)$ be \deviations{} that induce $f$. Let $\hat{\mathcal{C}}$ be an arbitrary directed cycle in $\hat{G}(f)$. We need to show that the cost of $\hat{\mathcal{C}}$ is non-negative, i.e., 
%$$
%c(\hat{\mathcal{C}}) 
%= \sum_{a \in \hat{\mathcal{C}}} {c}(a) 
%= \sum_{a \in \hat{\mathcal{C}} \cap A} (l_a + \theta^{\max}_a) - \sum_{a \in \hat{\mathcal{C}} \cap \bar{A}} (l_a+ \theta^{\min}_a) \ge 0.
%$$
%Assume that $\hat{\mathcal{C}}$ consist of forward arcs only. Then we have $c(\hat{\mathcal{C}}) = \sum_{a \in \hat{\mathcal{C}}} (l_a + \theta^{\max}_a) \geq \sum_{a \in \hat{\mathcal{C}}} (l_a + \theta^{\min}_a) \geq 0$, where the last inequality holds because of Assumption \ref{ass:nonnegative}.
%Next, assume that $\hat{\mathcal{C}}$ contains at least one reversed arc. Let $\mathcal{C} = (b_1, \dots, b_q, b_1)$ be the cycle that we obtain from $\hat{\mathcal{C}}$ by replacing all reversed edges $(v, u) \in \hat{\mathcal{C}}$ by $(u,v)$.

%\gsrem{I replaced $\mathcal{C}$ with $\hat{B}$ everywhere, since $\hat{B}$ is how this cycles is defined in the beginning of the proof.}
Summing up these inequalities for $0 \leq w \leq p-1$, we obtain 
%\gsrem{Directly apply the aggregated result to $\hat{B}$, hope that's okay..}
$$
\sum_{(v,u) \in \hat{B} \cap \bar{A}} l_{(u,v)} + \delta_{(u,v)} \leq \sum_{a \in \hat{B} \cap A} l_a + \delta_a,
$$
since all the contributions of the path $P_l$ cancel out. 
%which, for the cycle $\mathcal{C}$ in the graph $\hat{G}(f)$, is equivalent to
%$\sum_{a \in \mathcal{C} \cap A} (l_a + \delta_a) - \sum_{a \in \mathcal{C} \cap \bar{A}^+} (l_a + \delta_a) \geq 0.$
Now using the definition of a $\theta$-\deviation{}, we find 
$$
\sum_{a \in \hat{B}  \cap A} (l_a + \theta^{\max}_a) - \sum_{(v,u) \in \hat{B} \cap \bar{A}} (l_{(u,v)} + \theta^{\min}_{(u,v)}) \geq \sum_{a \in \hat{B}  \cap A} (l_a + \delta_a) - \sum_{(v,u) \in \hat{B}  \cap \bar{A}} (l_{(u,v)} + \delta_{(u,v)}) \geq 0.
$$
We have shown that $\hat{B} $ has non-negative cost. Note that $\hat{B} $ as zero cost if all the arcs on the cycle are reversed. 

\bigskip
\noindent For the other direction of the proof, consider the set $\mathcal{F}(\theta)$ of $\theta$-\deviations{} $\delta \in \Delta(\theta)$ that induce $f = (f^i_a)_{i \in [k], a \in A}$ (see also \cite{Lin2011,Roughgarden2006}):
\begin{align}
\mathcal{F}(\theta) = \{(\delta_a)_{a \in A}\ \  \big| \ \ 
\pi_{i,v} - \pi_{i,u} 
 \leq l_a(f_a) + \delta_a(f_a)  & \quad \forall a = (u,v) \in A, \forall i \in [k] \notag \\
\pi_{i,v} - \pi_{i,u} 
 =  l_a(f_a) + \delta_a(f_a)  & \quad \forall a = (u,v) \in A^+_i, \forall i \in [k] \notag \\
\theta^{\min}_a(f_a) 
 \leq  \delta_a(f_a) \leq \theta^{\max}_a(f_a) & \quad \forall a \in A \}.
 \label{eqn:lp_toll}
\end{align}
That is, $f$ is $\theta$-inducible if and only if (\ref{eqn:lp_toll}) has a feasible solution.
Now suppose that $\hat{G}(f)$ does not contain a cycle of negative cost. Then we can determine the shortest path distance $\delta_u$ from $s$ to every node $u \in V$. We define $\pi_{i,u} := \pi_u$ for all $u \in V$ and $i \in [k]$. Furthermore, for $a = (u,v) \in A$, we define $\delta_a := \max\{\theta^{\min}_a, \pi_v - \pi_u - l_a\}$. We will now show that $\delta$ induces $f$ by showing that we have constructed a feasible solution for (\ref{eqn:lp_toll}). First of all, for all $i \in [k]$ and $a \in A \setminus A^+_i$, we have $\delta_a \geq \pi_v - \pi_u - l_a$, which is equivalent to $\pi_{i,v} - \pi_{i,u} \leq l_a + \delta_a$. Secondly, if $a = (u,v) \in A^+_i$, then $\pi_u - \pi_v \leq -l_a - \theta^{\min}_a$ (which we derive using the reversed arc $(v,u)$). But this is equivalent to $\pi_{i,v} - \pi_{i,u} - l_a \geq \theta^{\min}_a$. We can conclude that $\delta_a = \pi_{i,v} - \pi_{i,u} - l_a$. Furthermore, we clearly have $\delta_a \geq \theta^{\min}_a$. Lastly, for all $a = (u,v) \in A$ we have $\pi_v - \pi_u \leq l_a + \theta^{\max}_a$ which is equivalent to $\pi_v - \pi_u - l_a \leq \theta^{\max}_a$. Combining this with the trivial inequality $\theta^{\min}_a \leq \theta^{\max}_a$ we can conclude that $\delta_a \leq \theta^{\max}_a$. This completes the proof. \label{thm:char_proof}\qed 
\end{proof}

\begin{remark}\label{rem:multi}
Consider the graph $G = (V,A)$ in Figure \ref{fig:multi_fail} and suppose that $r_1 = r_2 = 1$. Then the flow $f$ that routes one unit of flow over both paths $(s_1,v_1,1,2,t_1)$ and $(s_2,v_2,3,4,t_2)$ is feasible and inducible (take $\delta = 0$). However, looking at the graph $\hat{G}(f)$, we see the negative cost cycle $(1,4,3,2,1)$ (by using the reversed arcs of $(1,2)$ and $(3,4)$).
\qed
\end{remark}

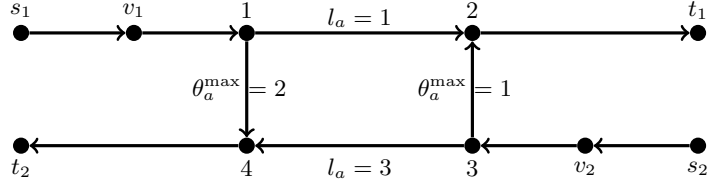
\begin{figure}[h!]
\centering
\begin{tikzpicture}
%\useasboundingbox (-8,-2) rectangle (5.5,3);   
\begin{scope}%[xshift=-3cm,yshift=-1cm] 

%Examples
%\node at (-1,2) [circle,fill=black] (s) {}; 
%\node at (0,2) [circle,fill=black] (1) {}; 
%\draw (s) edge[->, very thick, red, bend left] (1);

\node at (-3,1)  [circle,fill=black, scale = 0.7] (s1) {}; 
\node at (-1.5,1)  [circle,fill=black, scale = 0.7] (v1) {}; 
\node at (0,1)   [circle,fill=black, scale = 0.7] (1) {}; 
\node at (3,1)   [circle,fill=black, scale = 0.7] (2) {}; 
\node at (6,1)   [circle,fill=black, scale = 0.7] (t1) {}; 
\node at (-3,-0.5) [circle,fill=black, scale = 0.7] (t2) {}; 
\node at (0,-0.5)  [circle,fill=black, scale = 0.7] (4) {}; 
\node at (3,-0.5)  [circle,fill=black, scale = 0.7] (3) {}; 
\node at (4.5,-0.5)  [circle,fill=black, scale = 0.7] (v2) {}; 
\node at (6,-0.5)  [circle,fill=black, scale = 0.7] (s2) {};

\node at (-3,1.3) [] {$s_1$};
\node at (-1.5,1.3) [] {$v_1$};
\node at (0,1.3) [] {$1$};
\node at (3,1.3) [] {$2$};
\node at (6,1.3) [] {$t_1$};
\node at (-3,-0.8) [] {$t_2$};
\node at (0,-0.8) [] {$4$};
\node at (3,-0.8) [] {$3$};
\node at (4.5,-0.8) [] {$v_2$};
\node at (6,-0.8) [] {$s_2$};

\node at (1.5,1.2) [] {$l_a = 1$}; 
\node at (1.5,-0.8) [] {$l_a = 3$}; 
\node at (-0.1,0.25) [] {$\theta^{\max}_a = 2$}; 
\node at (2.9,0.25) [] {$\theta^{\max}_a = 1$}; 

\draw (s1) edge[->, very thick] (v1) ;
\draw (v1) edge[->, very thick] (1) ;
\draw (1)  edge[->, very thick] (2) ;
\draw (2)  edge[->, very thick] (t1);
\draw (s2) edge[->, very thick] (v2);
\draw (v2) edge[->, very thick] (3);
\draw (3)  edge[->, very thick] (4);
\draw (4)  edge[->, very thick] (t2);
\draw (3)  edge[->, very thick] (2);
\draw (1)  edge[->, very thick] (4);

\end{scope}
\end{tikzpicture}
\caption{All the values of $l_a, \ \theta^{\min}_a$ and $\theta^{\max}_a$ that are not explicitly stated are zero.} 
\label{fig:multi_fail}
\end{figure}

\subsection{Proof of Lemma \ref{lem:cycles}}

\begin{rtheorem}{Lemma}{\ref{lem:cycles}}
Let $x$ be $\theta$-inducible and let $X_i$ be a flow-carrying $(s,t_i)$-path for commodity $i \in [k]$ in $G$. Let $\chi$ and $\psi$ be any $(s,t_i)$-path and $(t_i,s)$-path in $\hat{G}(x)$, respectively. Then 
\begin{align*}
\sum_{a \in X_i} l_a(x_a) + \theta_a^{\min}(x_a) & \leq \sum_{a \in \chi \cap A} l_a(x_a) + \theta_a^{\max}(x_a) - \sum_{a \in \chi \cap \bar{A}} l_a(x_a) + \theta_a^{\min}(x_a) \\
\sum_{a \in X_i} l_a(x_a) + \theta_a^{\max}(x_a) & \geq \sum_{a \in \psi \cap \bar{A}} l_a(x_a) + \theta_a^{\min}(x_a) - \sum_{a \in \psi \cap A} l_a(x_a) + \theta_a^{\max}(x_a).
\end{align*}
\end{rtheorem}

We need the following proposition to prove Lemma~\ref{lem:cycles}. 
\begin{proposition}
Let $G = (V, A)$ be a non-empty, directed multigraph with the property that $\delta^-(v) = \delta^+(v)$ for all $v \in V$. Then $G$ is the union of arc-disjoint directed (simple) cycles $C_1,\dots,C_j$, such that $\bigcup_{j'} V(C_{j'}) = V(C)$ and $\bigcup_{j'} A(C_{j'}) = A(C)$. 
\label{prop:cycles}
\end{proposition}
\begin{proof}
If $G$ is non-empty then it is clear that we can always find a (simple) directed cycle $C$ in $G$. Removing the arcs of this cycle leads to the graph $G \setminus C := (V, A \setminus A(C))$ that also satisfies $\delta^-(v) = \delta^+(v)$ for all $v \in V$ (note that if there are multiple arcs between two nodes, we only remove the copy on the cycle). \qed
\end{proof} 

\begin{proof}[Lemma \ref{lem:cycles}]
Since $X_i$ is a flow-carrying path, we know that for every $a = (u,v) \in X_i$, we have a reversed arc $(v,u) \in \hat{A}$ in $\hat{G}$. Furthermore, any $(s,t_i)$-path in $\hat{G}$ can consist of both forward as well as reversed arcs. Let $\hat{H}$ be the graph consisting of the reversed path of $X_i$ (say $X_i'$), and the path $\chi$, where we add a copy of an arc if it is used by both paths (i.e., $\hat{H}$ can be a multigraph). Note that $\hat{H}$ satisfies the conditions of Proposition \ref{prop:cycles}, since it is the union of an $(s,t_i)$-path and a $(t_i,s)$-path. Therefore, the graph $\hat{H}$ is the union of arc-disjoint directed cycles $C_1,\dots, C_j$ for some $j$. Now, we apply Theorem \ref{thm:induc} to all these cycles and obtain
$$
\sum_{a \in A \cap C_{j'}} (l_a(x_a) + \theta^{\max}_a)(x_a) - \sum_{a \in \bar{A} \cap C_{j'}} (l_a(x_a) + \theta^{\min}_a)(x_a) \ \geq \ 0
$$
for all $j' = 1,\dots,j$. Adding up these inequalities then gives the desired result. The second inequality can be proved similarly (by applying the first argument in the opposite direction of the cycle). \qed
\end{proof}

\subsection{Proof of Lemma~\ref{lem:alt_path_tree}}

\begin{rtheorem}{Lemma}{\ref{lem:alt_path_tree}}
%Let $\mathcal{I}$ be a common source multi-commodity instance.
Let $z$ and $x$ be feasible flows and let $Z$ and $X$ be a partition of $A$ as in Definition \ref{def:alt_path}. Then there exists an alternating path tree. 
\end{rtheorem}
\begin{proof}%[Lemma \ref{lem:alt_path_tree}]
Let $G' = (V',A')$ be the graph defined by $V = V 
\cup \{t\}$ and $A' = A \cup \{(t_i,t) : i \in [k]\}$. Let $x', z'$ be the flows defined by 
$$
x_a' = \left\{ \begin{array}{rl} x_a & \text{ for } a = (u,v) \in A \\
r_i & \text{ for } a = (t_i,t) \text{ with } i \in [k]
\end{array}\right. \ 
\text{ and } \ 
z_a' = \left\{ \begin{array}{rl} z_a & \text{ for } a = (u,v) \in A \\
r_i & \text{ for } a = (t_i,t) \text{ with } i \in [k]
\end{array}\right.
$$
Then $x'$ and $z'$ are feasible $(s,t)$-flows in $G'$. We can write $A = Z' \cup X'$ with $Z' = Z \cup \{(t_i,t) : i \in [k]\}$ and $X'$ having the same properties as $Z$ and $X$ in $G$ (which follows from $x_a' = z_a' = r_i > 0$ for all $a = (t_i,t)$).\\
\indent We can now apply the same argument as in the proof of Lemma 4.5 in  \cite{Nikolova2015} of which we will give a short summary (for sake of completeness). For any $s$-$t$ cut defined by $S \cup V'$ with $s \in S$ we claim that we can cross $S$ with an arc in $Z'$, or a reversed arc in $X'$. Suppose that this would not be the case, i.e., all arcs into $S$ are in the set $Z'$ and all the outgoing arcs of $S$ are in $X'$. Let $x_{Z'}$ and $z_{Z'}$ be the total incoming flows from $S$, and $x_{X'}$ and $z_{X'}$ the total outgoing flows from $S$ (for resp. flows $x$ and $z$). From the definition of $Z'$ it follows that $x_{Z'} \leq z_{Z'}$. From conservation of flow it follows that $x_{X'} - x_{Z'} = z_{X'} - z_{Z'}$. Combining these two observations, we find that $x_{X'} \leq z_{X'}$. However, by definition of $X'$, we have $x_{X'} > z_{X'}$ (since we removed all arcs $a$ with $z_a = x_a = 0$). We find a contradiction.\\
\indent Having proved the claim that we can always cross with an arc in $Z'$ or a reversed arc in $X'$, we can now easily construct a spanning tree $\pi'$ consisting of alternating paths, by starting with the cut $(S, G \setminus S)$ given by $S = \{s\}$.\\
\indent Note that $t$ cannot be an interior point of $\pi'$, since $t$ is only adjacent to incoming arcs of the set $Z'$. This means that if we remove $(t_j,t)$ from $\pi'$ (where $j$ is the index for which $(t_j,t)$ is in the tree $\pi'$), we have found an alternating path tree $\pi$ for the graph $G$, under the flows $x$ and $z$. \qed 
\end{proof}

%\textit{Proof of Lemma \ref{lem:sc_bound}.} We know that $l_q(x) + \delta_q(x) = l_p(x) + \delta_p(x)$ for all flow-carrying paths $q \in \mathcal{P}$. But then
%\begin{eqnarray}
%\sum_{q \in \mathcal{P}_i} x_q^i l_q(x) &=& \sum_{q \in \mathcal{P}_i : x_q^i > 0} x_q^i l_q(x)\nonumber \\
% &=& \sum_{q \in \mathcal{P}_i : x_q^i > 0} x_q^i(l_p(x) + \delta_p(x) - \delta_q(x))\nonumber \\
% & = & r_i l_p(x) + \sum_{q \in \mathcal{P}_i : x_q^i > 0} x_q^i(\delta_p(x) - \delta_q(x)) \nonumber \\
% &\leq & r_i l_p(x). \nonumber
%\end{eqnarray} \qed 

\subsection{Proof of Theorem~\ref{thm:general_multi_common_source}(ii)}

We need the following lemma and proposition for the proof of Theorem~\ref{thm:general_multi_common_source}(ii).

\begin{lemma}
Let $-1 < \alpha \leq 0 \leq \beta$ be fixed. Then $f$ is inducible with an $(\alpha,\beta)$-\deviation{} if and only if it is inducible with a $(0,\frac{\beta - \alpha}{1 + \alpha})$-\deviation{}. 
\label{lem:alphabeta_induc}
\end{lemma} 
\begin{proof}
Let $f$ be inducible for some $\alpha l \leq \delta \leq \beta l$, and for $a \in A$, write $\delta_a(f_a) = d_al_a(f_a)$. Without loss of generality we may assume that $\delta_a(x) = d_a l_a(x)$ (since by definition $d_al_a(x)$ also induces $f$). From the equilibrium conditions, we know that 
$$
\forall i \in [k], \forall P \in \mathcal{P}_i, f^{\delta}_P > 0: \ \ \ \ \ \ \sum_{a \in P} l_a(f_a) + \delta_a(f_a) \leq \sum_{a \in P'} l_a(f_a) + \delta_a(f_a)  \ \ \forall P' \in \mathcal{P}_i.
$$
%where $f_a = \sum_{i \in [k]} f_a^i$. 
This is equivalent to $\forall i \in [k], \forall P \in \mathcal{P}_i, f^{\delta}_P > 0:$
$$
\sum_{a \in P} \left(1 + \frac{d_a - \alpha}{1 + \alpha}\right)l_a(f_a) \leq \sum_{a \in P'} \left(1 + \frac{d_a - \alpha}{1 + \alpha}\right)l_a(f_a)  \ \ \forall P' \in \mathcal{P}_i
$$
which can be seen by writing 
$$l_a(f_a) + \delta_a(f_a) = (1 + d_a)l_a(f_a) = (1 + \alpha + d_a - \alpha)l_a(f_a),$$
and then dividing the inequality by $1 + \alpha$. We then see that $\delta'$, defined by $\delta_a'(x) = \frac{d_a - \alpha}{1 + \alpha} l_a(x)$ for all $a \in A$ and $x \geq 0$, also induces $f$, since
$$
\alpha l_a(x) \leq d_al_a(x) \leq \beta l_a(x)  \ \ \ \Leftrightarrow \ \ \
0 \leq \frac{d_a - \alpha}{1 + \alpha} l_a(x) \leq \frac{\beta - \alpha}{1 + \alpha} l_a(x).
$$
%This completes the proof, since we work with equivalences.  
\qed
\end{proof} 
\begin{proposition}
Let $z = f^0$ be a Nash flow for a multi-commodity instance with a common source. Let $v \in V$ and let $i,j \in [k]$ be two commodities for which there exist flow-carrying $(s,v)$-paths $P_1 \in \mathcal{P}_i$ and $P_2 \in \mathcal{P}_j$, respectively. 
%\gsrem{Pieter: please see revised suggestion.}
Then there exists a feasible Nash flow $\bar{z}$ with $\bar{z}_a = z_a$ for all $a \in A$ such that both paths $P_1, P_2$ are flow-carrying for commodity $i$, and both paths $P_1,P_2$ are flow-carrying for commodity $j$, i.e., we have $\bar{z}_{P_1}^i, \bar{z}_{P_2}^i, \bar{z}_{P_1}^j, \bar{z}_{P_2}^j > 0$.
%\pkrem{Much better; now only the double superscript is disturbing; can we choose another letter? Or $\bar{z}$ instead of $z'$? Just for typographical beauty.. ;-) ACTUALLY JUST CHANGED IT TO $\bar{z}$. Okay?}
\label{prop:interchange}
\end{proposition}
\begin{proof}
Intuitively, we shift an $\epsilon$ amount of flow of commodity $i$ to path $P_2$ and an $\epsilon$ amount of flow of commodity $j$ to path $P_1$. Formally, choose $\epsilon > 0$ small enough such that $z_{P_1}^i - \epsilon, z_{P_2}^j - \epsilon > 0$. We define 
$$
\bar{z}_{P}^l = \left\{\begin{array}{ll} z_{P_1}^i - \epsilon &  \ \ \text { if } P = P_1 \text{ and } l = i\\
z_{P_1}^j + \epsilon & \ \ \text { if } P = P_1 \text{ and } l = j\\
z_{P_2}^i + \epsilon & \ \ \text { if } P = P_2 \text{ and } l = i\\
z_{P_2}^j - \epsilon & \ \ \text { if } P = P_2 \text{ and } l = j\\
\end{array}\right.
$$
and let all the other flow-carrying paths remain unchanged.
It then immediately follows that $z_a = \bar{z}_a$ for all $a \in A$, and in the resulting feasible flow $\bar{z}$, both commodities $i$ and $j$ are flow-carrying for both paths $P_1$ and $P_2$. The feasibility of $\bar{z}$ follows because both commodities have the same source. Moreover, the common source also implies that if $z$ is a Nash flow, then $\bar{z}$ is also a Nash flow (since commodity $i$ implies that $l_{P_1}(z) \leq l_{P_2}(z)$, and commodity $j$ implies that $l_{P_2}(z) \leq l_{P_1}(z)$). \qed 
\end{proof}

\begin{proof}[Theorem \ref{thm:general_multi_common_source}(ii)]
By Lemma \ref{lem:alphabeta_induc} we can assume without loss of generality that $\theta^{\max}_a = \frac{\beta - \alpha}{1 + \alpha} l_a$ and $\theta^{\min}_a = 0$ for all $a \in A$. Furthermore, with $A_{ij}$ we denote the \emph{$j$-th segment} of $\pi_i$, $j = 1,\dots,\eta_i$, consisting of consecutive arcs in $Z$. Using Theorem \ref{lem:main_bound} and the definition of $A_{ij}$, we obtain
\begin{eqnarray}
C(x) & \leq & C(z) + \frac{\beta - \alpha}{1 + \alpha} \sum_{i \in [k]} r_i \sum_{a \in P \cap \pi_i} l_a(z_a) \ \nonumber \\
& \leq & C(z) + \frac{\beta - \alpha}{1 + \alpha} \sum_{i \in [k]} r_i \left(\eta_i \cdot\max_{j = 1,\dots,\eta_i}\sum_{a \in A_{ij}} l_a(z_a) \right)\nonumber 
%& \leq &  C(z) + \frac{\beta - \alpha}{1 + \alpha} \sum_{i \in [k]} r_i \eta_i C(z) \nonumber \\
%& = & C(z) + \frac{\beta - \alpha}{1 + \alpha} \left( \sum_{i \in [k]} r_i \eta_i\right) C(z) \nonumber
%& \leq & C(z) + \frac{\beta - \alpha}{1 + \alpha} \sum_{i \in [k]} \eta_i C(z) \nonumber \\
%& = & C(z) + \frac{\beta - \alpha}{1 + \alpha} \left( \sum_{i \in [k]} \eta_i\right) C(z) \nonumber 
\end{eqnarray}
Note that it now suffices to show that $\sum_{a \in A_{ij}} l_a(z_a) \leq C(z)$ for all $j = 1,\dots,\eta_i$ and $i \in [k]$.

We prove below that, for a fixed section $A_{ij}$, there exists a commodity $w \in [k]$ such that every $a \in A_{ij}$ is flow-carrying for commodity $w$ (note that $w$ and $i$ can be different). This allows us to assume that $A_{ij}$ is contained in some flow-carrying path $l_{w} \in \mathcal{P}_w$ (by choosing a suitable path decomposition of $z$ for commodity $w$). We then obtain that 
$\sum_{a \in A_{ij}} l_a(z_a) \leq l_{w}(z) \leq C(z)$  since $r_i \geq 1$.
%\textcolor{red}{$^*$}. 
Recall that $C(z) = \sum_{i \in [k]} r_i l_{Z_i}(z)$, where $Z_i \in \mathcal{P}_i$ is an arbitrary flow-carrying path for commodity $i \in [k]$. 

We will now prove the above claim. Fix a section $A_{ij}$ and let $a_1 = (u,v)$ and $a_2 = (v,w)$ be two consecutive arcs that are flow-carrying for commodities $w_1$ and $w_2$ in $z$, respectively. This implies that there are flow-carrying $(s,v)$-paths $W_1$ and $W_2$ such that $W_1$ is flow-carrying for $w_1$, and $W_2$ for $w_2$. The existence of $W_1$ is clear, and the existence of $W_2$ follows from flow-conservation applied to commodity $w_2$ (since flow is leaving node $v$ for that commodity). But then, by Proposition \ref{prop:interchange}, we may assume that $a_1$ is also flow-carrying for commodity $w_2$. Applying this argument repeatedly, starting with the last two arcs on $A_{ij}$ and working to the front, we can show that the whole section $A_{ij}$ is flow-carrying for a commodity that is flow-carrying on the last arc of $A_{ij}$. \qed 
\end{proof}

%\textcolor{red}{$^*$Regarding normalization of $r_1 = 1$. Without normalization, we get
%$$
%\sum_{a \in A_{ij}} l_a(z_a) \leq l_{w}(z) \leq \frac{r_i}{r_1} l_w(z)\leq \frac{C(z)}{r_1}
%$$
%using that $r_1 \leq r_2 \leq ..$. Then in the bound we would get $C(x) \leq \left(1 + (\beta - \alpha)/(1 + \alpha)\sum_{i \in [k]} \frac{r_i}{r_1} \eta_i\right) C(z)$}.

\newpage

\section{Omitted material of Section 4}

\subsection{Proof of Theorem~\ref{thm:lower_multi_odd}}

\begin{rtheorem}{Theorem}{\ref{thm:lower_multi_odd}}
There exist common source two-commodity instances $\mathcal{I}$ such that 
$$
\dr(\mathcal{I}, (\alpha, \beta)) 
\ge 
\left\{ \begin{array}{ll}
1 + (\beta - \alpha)/(1+\alpha)\cdot (n-1)/2 \cdot r & \text{ for } n = 2m+1 \in \N_{\geq 5} \\
1 + (\beta - \alpha)/(1+\alpha)\cdot [(n/2-1)r + 1] & \text{ for } n = 2m \in \N_{\geq 4}.
\end{array}\right.
$$
\end{rtheorem}
\begin{proof}%[Theorem \ref{thm:lower_multi_odd}]
We first prove the claim for $n$ odd. 
Let $r \in \R_{\geq 1}$ and $n = 2m + 1 \in \N_{\geq 5}$. We modify the graph $G^m$ by adding one extra node $t_2$ (the node $t$ will be referred to as $t_1$ from here on). We add the arcs $(s,t_2)$ and $(t_2,t_1)$ (see the dotted arcs in Figure \ref{fig:braess_5}). We take one commodity with sink $t_1$ and $r_1 = 1$, and one commodity with sink $t_2$ and demand $r_2 = r - 1$. Note that the latter commodity only has one $(s,t_2)$-path.

The pairs $(l_a^m(g),\delta_a^m(g))$, for all $a$ except $(s,t_2)$ and $(t_2,t_1)$, are defined as in Example \ref{exmp:tight_pra}, but with $y$ a non-decreasing, non-negative, continuous function satisfying $y_m(1/m) = 0$ and $y_m((1-\epsilon_m)/(m-1)) = \beta$, where we choose $0 < \epsilon_m < 1/m$\, so that $1/m < (1 - \epsilon_m)/(m - 1)$. For $a = (s,t_2)$, we take $(l_a^m(g),\delta_a^m(g)) = (y^*_m(x'),0)$, where $y^*$ is a non-decreasing, non-negative, continuous function satisfying $y_m^*(r-1) = 0$ and $y_m^*(r-1 + \epsilon_m) = \beta$. For $a = (t_2,t_1)$ we take $(l_a^m(g),\delta_a^m(g)) = (1,0)$. See Figure~\ref{fig:braess_5_multi} for an example.

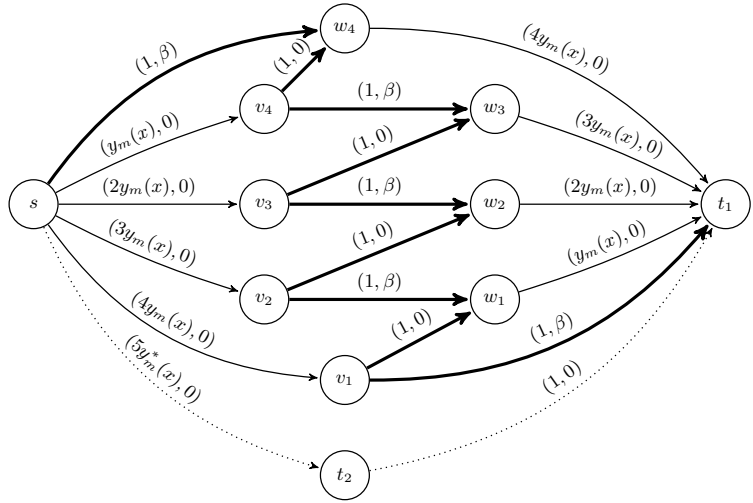
\begin{figure}[h!]
\centering
\scalebox{0.8}{
\begin{tikzpicture}[
  ->,
  >=stealth',
  shorten >=1pt,
  auto,
  %node distance=2cm,
  semithick,
  every state/.style={circle,fill=white,radius=0.2cm,text=black},
]
r = 1
\begin{scope}
  \node[state]  (s)               					 {$s$};
  \node[state]  (v3) [right=3cm of s] 				 {$v_3$};
  \node[state]  (v4) [above=0.75cm of v3]  			 {$v_4$};
  \node[state]  (v2) [below=0.75cm of v3]			     {$v_2$};
  \node[state]  (v1) [below right=0.75cm and 0.75cm of v2] {$v_1$};
  \node[state]  (w4) [above right=0.75cm and 0.75cm of v4] {$w_4$};
  \node[state]  (w3) [right=3cm of v4] 				 {$w_3$};
  \node[state]  (w2) [right=3cm of v3] 				 {$w_2$};       
  \node[state]  (w1) [right=3cm of v2] 				 {$w_1$}; 
  \node[state]  (t)  [right=3cm of w2] 				 {$t_1$};
  \node[state]  (t2) [below=0.75cm of v1] 			 {$t_2$};
  
\path[every node/.style={sloped,anchor=south,auto=false}]
(s) edge[line width=1.5pt, bend left=25] 	node {$(1,\beta)$} (w4)            
(s) edge[bend left=5]  						node {$(y_m(x),0)$} (v4)
(s) edge               						node {$(2y_m(x),0)$} (v3)
(s) edge[bend right=5] 						node {$(3y_m(x),0)$} (v2)
(s) edge[bend right=25]						node {$(4y_m(x),0)$} (v1)
(v4) edge[line width=1.5pt]             	node {$(1,0)$} (w4)
(v4) edge[line width=1.5pt]         		node {$(1,\beta)$} (w3)
(v3) edge[line width=1.5pt]               node {$(1,0)$} (w3)
(v3) edge[line width=1.5pt]               node {$(1,\beta)$} (w2)
(v2) edge[line width=1.5pt]               node {$(1,0)$} (w2)
(v2) edge[line width=1.5pt]               node {$(1,\beta)$} (w1)
(v1) edge[line width=1.5pt]               node {$(1,0)$} (w1)
(w4) edge[bend left=25]node {$(4y_m(x),0)$} (t)            
(w3) edge[bend left=5] node {$(3y_m(x),0)$} (t)
(w2) edge              node {$(2y_m(x),0)$} (t)
(w1) edge[bend right=5]node {$(y_m(x),0)$} (t)
(v1) edge[line width=1.5pt, bend right=25]node {$(1,\beta)$} (t)
(s)  edge[bend right=25, dotted]node {$(5y^*_m(x),0)$} (t2)
(t2)  edge[bend right=25, dotted]node {$(1,0)$} (t);  
        
\end{scope}
\end{tikzpicture}}
\caption{The fifth (odd) Braess graph with $(l_a^5,\delta_a^5)$ on the arcs as defined above, where $t = t_1$. The thick edges indicate the alternating path $\pi_1$.} 
\label{fig:braess_5_multi}
\end{figure}

A Nash flow $z$ for this instance is given by routing $1/m$ units of flow over the paths $(s,w_{m-1},t_1),(s,v_1,t_1)$ and the paths in $\{(s,v_j,w_{j-1},t_1) : 2 \leq j \leq m-1\}$ for the first commodity, and $r - 1$ units of flow over $(s,t_2)$ for the second commodity. This claim is true since all the paths for the first commodity have latency one, as well as the paths $(s,v_j,w_j,t)$, for $1 \leq m \leq j$. This is also true for $(s,t_2,t_1)$. The latency for the other commodity is zero. We may conclude that $C(z) = 1$.

A Nash flow $x$ under \deviation{} $\delta$, as defined here, is given by, for the first commodity, routing $(1 - \epsilon_m)/(m-1)$ units of flow over the paths in $\{(s,v_j,w_j,t) : 1 \leq j \leq m - 1\}$, and $\epsilon_m$ units of flow over the path $(s,t_2,t_1)$. Note that the perceived latency on all these paths $p$ is $q_P(x) = 1 + \beta m$ (which is also the true latency, since all the \deviations{} are zero on the arcs of these paths). Using the same reasoning as in Example \ref{exmp:tight_pra} it can be seen that the perceived latency on the paths $P' = (s,v_j,w_{j-1},t)$, for $2 \leq j \leq m-1$, is also $q_{P'}(x) = 1 + \beta m$, from which we may conclude that $x$ is indeed a Nash flow under the \deviation{} $\delta$. We have$C(x) = 1 + \beta m + (r-1) \beta m = 1 + \beta rm$, since for the first commodity the (true) latency along every path is $1 + \beta m$, and for the other commodity the latency along $(s,t_2)$ is $\beta m$. 

We next prove the claim for $n$ even. 
Let $r \in \R_{\geq 1}$ and $n = 2m \in \N_{\geq 4}$. We use the same Braess graphs as in Example \ref{exmp:tight_pra}, without modifications. We introduce another commodity with demand $r_2 = r -1$, for which we choose $t_2 = v_1$.  %\footnote{Here we implicitly use that $n \geq 4$, since for $n = 2$, the resulting instance is a parallel-arc game on two arcs, which is actually a single-commodity instance.}.
We replace the pair $((m-1)y_m(x'),0)$ on $a = (s,v_1)$ by the pair $((m-1)y'_m(g),0)$ where $y'_m$ satisfies $y'_m(1/m + r-1) = 0$ and $y'_m(1/(m-1) + r-1) = \beta$. 
Note that the flows $x$ and $z$, as defined in Example \ref{exmp:tight_pra} with the extension that the second commodity uses the arc $(s,v_1)$ in both cases, still form feasible Nash flows for their respective \deviations{}. We obtain 
\begin{align*}
C(x) 
& = \sum_i \sum_{q \in \mathcal{P}_i} x_q^i l_q(x) 
 = 1 + \beta m + (r-1) (m-1) \beta \\
& = 1 + \beta  m + \beta (r-1) (m-1)  
  = (1 + \beta r m) - \beta(r-1).
%& = & \Omega\left(1 + \frac{\beta - \alpha}{1 + \alpha}\cdot k\cdot m\right) \nonumber
\end{align*} \qed
%Especially, if we fix $k$, the bound $1 + \frac{\beta - \alpha}{1 + \alpha}\cdot k\cdot m$ is tight up to a constant (which is $\frac{\beta - \alpha}{1 + \alpha}(k-1)$). \qed
\end{proof}

\begin{remark}\label{rem:improve-2-com}
For two-commodity instances with $n$ even, we can actually improve the upper bound in Theorem \ref{thm:general_multi_common_source} to the lower bound stated in Theorem \ref{thm:lower_multi_odd}:
Suppose the upper bound of Theorem \ref{thm:general_multi_common_source} is tight. Then we need to have $\eta_1 = \eta_2 = n/2$. This means that the alternating path tree is actually a path (in the sense that all nodes are adjacent to at most two arcs of the alternating path tree) that alternates between arcs in $X$ and $Z$, starting and ending with an arc in $Z$ (see Figure \ref{fig:braess_5}). However, because $t_1 \neq t_2$ this means that at least one of the two commodities has no more than $n/2 - 1$ arcs in $Z$, which is a contradiction. 
\end{remark}

\subsection{Proof of Theorem~\ref{thm:multi_general}}

\begin{rtheorem}{Theorem}{\ref{thm:multi_general}}
For every $p = 2q + 1 \in \N$, there exists a two-commodity instance $\mathcal{I}$ whose size is polynomially bounded in $p$ such that
%$|V|,|A| \in \mathcal{O}(p)$, for which 
$
\dr(\mathcal{I}, (\alpha, \beta))
\geq 1 + \beta F_{p+1} \approx 1 + 0.45 \beta \cdot \phi^{p+1},
$
where $F_p$ is the $p$-th Fibonacci number 
%$c\approx 0.4472$ 
and $\phi \approx 1.618$ is the golden ratio.
%\footnote{Remember that $F_p = F_{p-2} + F_{p-1}$ for $p \geq 2$, where $F_0 = 0$ and $F_1 = 1$.}
\end{rtheorem}

\bigskip
Our proof of Theorem~\ref{thm:multi_general} is based on the following graph, which was used by Lin et al. \cite{Lin2011}.

%\newpage

\begin{figure}[t!]
\centering
\scalebox{0.85}{
\begin{tikzpicture}[
  ->,
  >=stealth',
  shorten >=0.5pt,
  auto,
  %node distance=2cm,
  semithick,
  every state/.style={circle,fill=white,minimum size=0.3cm,text=black},
]
\begin{scope}
  %Horizontal nodes
  \node[state]  (s1)               					 				{$s_1$};
  \node[state]  (e)  [right=0.7cm of s1]				 			{$e$};
  \node[state]  (w1) [right=0.7cm of e]             				{$w_1$};
  \node[state]  (v1) [right=0.7cm of w1]             				{$v_1$};
  \node[state]  (v2) [right=0.7cm of v1]            				{$v_2$};
  \node[state]  (v3) [right=0.7cm of v2]            			    {$v_3$};
  \node[state]  (v4) [right=0.7cm of v3]           					{$v_4$};
  \node[state]  (v5) [right=0.7cm of v4]           					{$v_5$};
  \node[state]  (v6) [right=0.7cm of v5]           					{$v_6$};
  \node[state]  (v7) [right=0.7cm of v6]           					{$v_7$};
  \node[state]  (t1) [right=0.7cm of v7]          					{$t_1$};
  %Vertical nodes
  \node[state]  (w2) [above=0.7cm of w1]             				{$w_2$};
  \node[state]  (w3) [above=0.7cm of w2]             				{$w_3$};
  \node[state]  (w4) [above=0.7cm of w3]             				{$w_4$};
  \node[state]  (w5) [above=0.7cm of w4]             				{$w_5$};
  \node[state]  (w6) [above=0.7cm of w5]             				{$w_6$};
  \node[state]  (w7) [above=0.7cm of w6]             				{$w_7$};
  \node[state]  (t2) [above=0.7cm of w7]             				{$t_2$};
  \node[state]  (w0) [below=0.7cm of w1]             				{$w_0$};
  \node[state]  (s2) [below=0.7cm of w0]             				{$s_2$};

\path[every node/.style={sloped,anchor=south,auto=false}]
%Horizontal edges          
(s1) edge 				   node {$(1,\beta)$} (e)
(e) edge				   node {} (w1)
(w1) edge				   node {} (v1)
(v1) edge				   node[below=0.1cm] {$\beta g^1_\delta$} (v2)
(v2) edge				   node {} (v3)
(v3) edge				   node[below=0.1cm] {$\beta g^3_\delta$} (v4)
(v4) edge				   node {} (v5)
(v5) edge				   node[below=0.1cm] {$\beta g^5_\delta$} (v6)
(v6) edge				   node {} (v7)
(v7) edge				   node {} (t1)
%Vertical edges
(s2) edge 				   node {} (w0)
(w0) edge 				   node {$\beta g^1_\delta$} (w1)
(w1) edge 				   node {} (w2)
(w2) edge 				   node {$\beta g^2_\delta$} (w3)
(w3) edge 				   node {} (w4)
(w4) edge 				   node {$\beta g^4_\delta$} (w5)
(w5) edge 				   node {} (w6)
(w6) edge 				   node {$\beta g^6_\delta$} (w7)
(w7) edge 				   node {} (t2)
%Diagonal edges
(s1) edge [bend right=10]  node {$1$} (w0)
(e) edge [bend left=10]	   node {} (w2)
(e) edge [bend left=10]	   node {} (w4)
(e) edge [bend left=10]	   node {} (w6)
(s2) edge [bend right=10]  node {} (v1)
(s2) edge [bend right=10]  node {} (v3)
(s2) edge [bend right=10]  node {} (v5)

(v2) edge 			       node {} (w2)
(w3) edge 			       node {} (v3)
(v4) edge 			       node {} (w4)     
(w5) edge 			       node {} (v5)
(v6) edge 			       node {} (w6)     
(w7) edge 			       node {} (v7);
\end{scope}
\end{tikzpicture}}
\caption{The graph $G^p$ for $p = 7$ (this is a reproduction of (Fig. 4, \cite{Lin2011})). The arc $a = (s_1,e)$ has $\delta_a = \beta$, whereas all the other arcs have $\delta_a = 0$.}
\label{fig:multi_exp}
\end{figure}
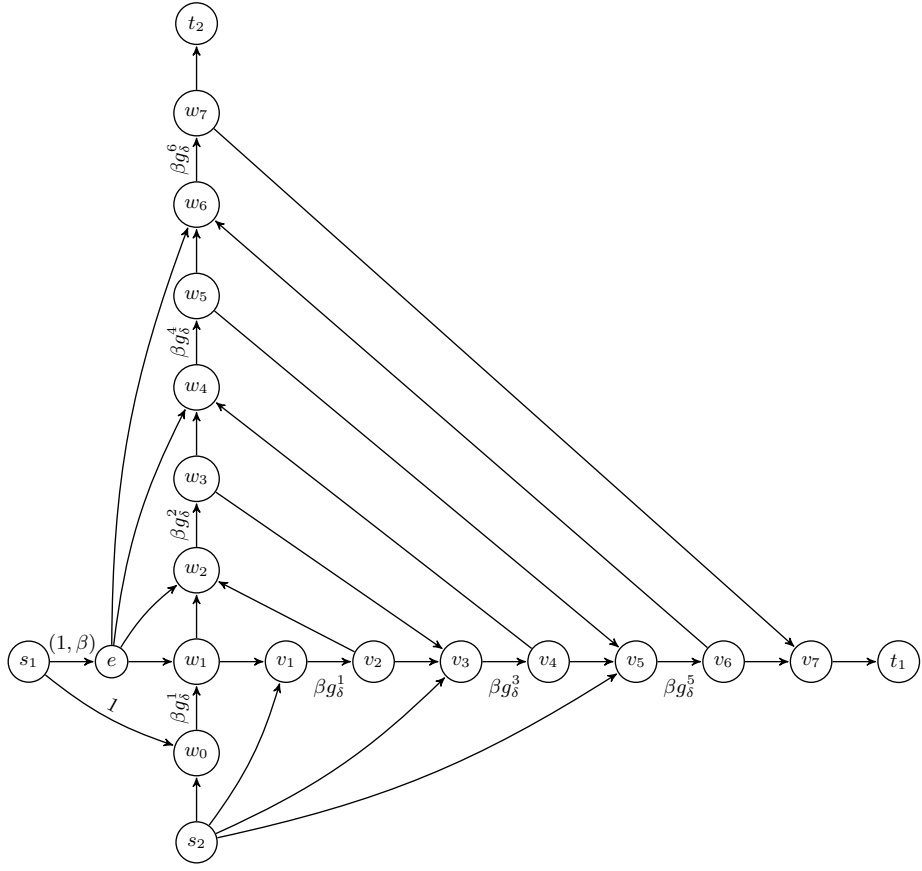

%\newpage

\begin{definition}[\!\!{\cite{Lin2011}}]
For $p = 2q + 1 \in \N$, the graph $G^p = (V^p,A^p)$ is defined by 
$$
V^p = \{s_1,s_2,t_1,t_2,e,w_0,\dots,w_p,v_1,\dots,v_p\},
$$
and $A^p = A(P_1^p) \cup A(P_2^p) \cup A^p_1 \cup A^p_2 \cup \{s_1,w_0\}$ where 
$$
P_1^p = (s_1,e,w_1,v_1,v_2,\dots,v_p,t_1) \text{ and }  P_2^p = (s_2,w_0,w_1,\dots,w_7,t_2)
$$ 
are the horizontal $(s_1,t_1)$-path and vertical $(s_2,t_2)$-path, respectively; see Figure \ref{fig:multi_exp}. Further,
$$
A^p_1 = \{(s_2,v_i) : i = 1,3,5,7,\dots,p-2 \} \cup \{(e,w_i) : i = 2,4,6,8,\dots,p-1\}
$$
and
$$
A^p_2 = \{(w_i,v_i) : i = 3,5,7,\dots,p\} \cup \{(v_i,w_i) : i = 2,4,6,8,\dots,p-1\}.  
$$
Lastly, the paths $T_i$ are denoted by
$$
T_i = \left\{ \begin{array}{ll}
(s_1,w_0,w_1,v_1,\dots,v_p,t_1) & i = 0 \\
(s_1,e,w_i,w_{i+1},v_{i+1},\dots,v_p,t_1) & i = 2,4,6,\dots,p-1 \\
(s_2,v_1,v_{i+1},w_{i+1},\dots,w_p,t_2) & i = 1,3,5,\dots,p
\end{array}\right.
$$
These paths can be seen as `shortcuts' for the paths $P_1$ and $P_2$.\qed
\label{def:gp_exp}
\end{definition}

\begin{proof}[Theorem \ref{thm:multi_general}]
We consider instances $(G^p,l^p,\delta^p,r^p)_{p = 1,3,5,7,\dots}$ with $G^p$ as in Definition \ref{def:gp_exp}. It is not hard to see that $|V^p|,\ |A^p| \in \mathcal{O}(p)$. 
The latency functions $l^p$ are given as follows:
$$
l^{p}_a(x') = \left\{ \begin{array}{ll}
\beta g_\delta^i(x') & \text{ for } a \in \{(v_i,v_{i+1}) : i = 1,3,5,\dots,p-2\}\\
\beta g_\delta^i(x') & \text{ for } a \in \{(w_i,w_{i+1}) : i = 0,2,4,6,\dots,p-1\}\\
1 & \text{ for } a \in \{(s_1,e),(s_1,w_0)\} \\
0 & \text{ otherwise.}
\end{array}\right.
$$
Here
$$
g_\delta^i(x') = 
\left\{ \begin{array}{ll} 0 & x' \leq 1 \\
h_\delta^i(x') & 1 \leq x' \leq 1 + \delta \\
F_i & x' \geq 1 + \delta, 
\end{array}\right.
$$
where $F_i$ is the $i$-th Fibonacci number, and $h_\delta^i(x')$ is some non-decreasing, non-negative, continuous function satisfying $h_\delta^i(1) = 0$ and $h_\delta^i(1 + \delta) = F_i$ (so that $g_\delta^i(x')$ is also non-decreasing, non-negative and continuous). Furthermore, we take $\delta_a = \beta$ for $a = (s_1,e)$ and $\delta_a = 0$ for all $a \in A \setminus \{(s_1,e)\}$. Finally, we have $r^p_1 = r^p_2 = 1$.

Let $z$ be the defined by sending one unit of flow over the paths $P_1$ and $P_2$. We claim that $z$ is a Nash flow with respect to the latencies $l^p$ and $C(z) = 1$. By construction, the latency along the path $P_1$ is $l_{P_1}(z) = 1$. It is not hard to see that any $(s_1,t_1)$-path has latency greater or equal than one (because every path for commodity $1$ uses either $(s_1,e)$ or $(s_1,w_0)$). For commodity $2$ the latency along $P_2$ is $l_{P_2}(z) = 0$, which is clearly a shortest path. This proves that $z$ is a Nash flow. Further, $C(z) = 1$.

We use Lemma \ref{lem:multi_exp} (given below) to describe a Nash flow $x$ with respect to the deviated latencies $l^p + \delta^p$. 
It follows that $C(x) = C(x)/C(z) \geq 1 + \beta F_{p-1} + \beta F_p = 1 + \beta F_{p+1}$. 
This concludes the proof (since $F_p \approx c \cdot \phi^p$ where $c \approx 0.4472$ and $\phi \approx 1.618$).
\qed
\end{proof} 

The following lemma is similar to Lemma 5.4, Lemma 5.5 and Lemma 5.6 in \cite{Lin2011}).

\begin{lemma} 
There exists a $\delta > 0$ and a feasible flow $x$ satisfying the following properties:
\begin{enumerate}[(i)]
\item $x_{a} \geq 1 + \delta$ for all $a \in \{(v_i,v_{i+1}) : i = 1,3,5,\dots,p-2\} \cup \{(w_i,w_{i+1}) : i = 0,2,4,6,\dots,p-1\}$.
\item $l_P(x) \geq 1 + \beta F_{p-1}$ for all $P \in \mathcal{P}_1$, with equality if and only if  $P = T_i$ for some $i = 2,4,6,\dots,p-1$.
\item $l_P(x) \geq \beta F_{p}$ for all $P \in \mathcal{P}_2$, with equality if and only if  $P = T_i$ for some $i = 1,3,5,\dots,p$.
\item $x$ is a Nash flow under the perceived latencies $l^p + \delta^p$.
\end{enumerate}
\label{lem:multi_exp}
\end{lemma}
\begin{proof}
The statements (i)--(iii) follow from Lemma 5.4, Lemma 5.5 and Lemma 5.6 in \cite{Lin2011}. The last statement is clearly true for commodity $2$ (since this commodity is not affected by the \deviation{} on arc $(s_1,e)$). For commodity $1$, all the flow-carrying paths $T_i$ have a perceived latency of $Q_{T_i}(x) = 1 + \beta(F_{p} + 1)$, and the perceived latency along any other $(s_1,t_1)$-path is greater or equal than that. The actual latencies along these paths are $l_{T_i}(x) = 1 + \beta F_{p-1}$ for $i = 2,4,6,\dots,p-1$, and $l_{T_0}(x) = 1 + \beta (F_{p-1}+1)$.
% (which yields claim $ii)$).
\qed
\end{proof}

\subsection{Proof of Theorem~\ref{thm:pra_model_results}}

%\gsrem{theorem statement changed}

\begin{rtheorem}{Theorem}{\ref{thm:pra_model_results}}
The Price of Risk Aversion for a common source multi-commodity instance $\mathcal{I}$ with non-negative, non-decreasing latency functions, variance-to-mean-ratio $\kappa > 0$ and risk-aversion parameter $\gamma \ge -1/\kappa$ is at most 
$$
\text{PRA}(\mathcal{I}, \gamma, \kappa) \le 
\begin{cases} 
1 - \gamma \kappa /(1 + \gamma \kappa ) \lceil (n-1)/2 \rceil r & \text{for $-1/\kappa < \gamma \leq 0$} \\
1 + \gamma \kappa \lceil (n-1)/2 \rceil r & \text{for $\gamma \ge 0$}.
\end{cases}
$$
Moreover, these bounds are tight for all $\kappa \geq 0$, $\gamma \in (-1/\kappa,\infty)$, $r \geq 1$, and $n = 2m+1 \in \N$. For $n = 2m \in \N$, we obtain
$$
\text{PRA}(\mathcal{I}, \gamma, \kappa) \geq  
\begin{cases}
(1 + \gamma \kappa r \lceil (n-1)/2 \rceil) - \gamma \kappa(r-1) & \text{ if $\gamma \geq 0$} \\
(1 - \frac{\gamma \kappa}{1 + \gamma \kappa} r \lceil (n-1)/2 \rceil) + \frac{\gamma \kappa}{1 + \gamma \kappa}(r-1) &  \text{if $-1/\kappa < \gamma \leq 0.$} 
\end{cases}
$$ 
In particular, for single-commodity instances we obtain tightness for all $n \in \N$.
\end{rtheorem}
\begin{proof}%[Theorem \ref{thm:pra_model_results}]
Recall from the discussion in Section~\ref{sec:pre} that the \deviations{} $\delta_a = \gamma v_a$ can be interpreted as $\theta$-\deviations{} with
$$
\theta^{\min}_a = \left\{ \begin{array}{ll} 0 & \text{ if } \gamma \geq 0\\ 
\gamma \kappa l_a & \text{ if } -1/\kappa < \gamma \leq 0 
\end{array}  \right. \ \ \ \ \text{ and } \ \ \ \ \theta^{\max}_a = \left\{ \begin{array}{ll} \gamma \kappa l_a & \text{ if } \gamma \geq 0\\ 
0 & \text{ if } -1/\kappa < \gamma \leq 0 .
\end{array}  \right. 
$$
Here, the restriction $\gamma > -1/\kappa$ is necessary to satisfy Assumption \ref{ass:nonnegative}.
The theorem now follows directly from Theorem~\ref{thm:general_multi_common_source}, Example \ref{exmp:tight_pra} and Theorem \ref{thm:lower_multi_odd}.
\qed
\end{proof}

\subsection{Proof of Theorem~\ref{thm:stability}}

\begin{rtheorem}{Theorem}{\ref{thm:stability}}
Let $\mathcal{I}$ be a common source multi-commodity instance with non-negative and non-decreasing latency functions $(l_a)_{a \in A}$. Let $f$ be a Nash flow with respect to $(l_a)_{a \in A}$ and let $\tilde{f}$ be a Nash flow with respect to slightly perturbed latency functions $(\tilde{l}_a)_{a \in A}$ satisfying 
$$
\sup_{a \in A,\; x \geq 0} \bigg| \frac{l_a(x) - \tilde{l}_a(x)}{l_a(x)} \bigg| \leq \epsilon
$$
for some small $\epsilon > 0$.
Then the relative error in social cost is $(C(\tilde{f}) - C(f))/C(f) \le 
2\epsilon/(1 - \epsilon) \lceil(n-1)/2 \rceil \cdot r 
= \mathcal{O}(\epsilon rn)$.
\end{rtheorem}

\begin{proof}
Note that the perturbation $l - \tilde{l}$ can be seen as an $(-\epsilon,\epsilon)$-\deviation{}. Theorem \ref{thm:general_multi_common_source} gives $C(\tilde{f})/C(f) \leq 1 + 2\epsilon/(1 - \epsilon) \lceil(n-1)/2 \rceil \cdot r$. This implies that the relative error in social cost is $(C(\tilde{f}) - C(f))/C(f) \leq 2\epsilon/(1 - \epsilon) \lceil(n-1)/2 \rceil \cdot r = \mathcal{O}(\epsilon rn)$ for small $\epsilon > 0$.
\qed
\end{proof}

\newpage

\section{Missing material of Section 5}

\subsection{Proof of Theorem~\ref{thm:smoothness}}

\begin{rtheorem}{Theorem}{\ref{thm:smooth-BPA}}
Let $\mathcal{L}$ be a set of non-negative, non-decreasing and continuous functions. Let $\mathcal{I}$ be a general multi-commodity instance with $(l_a)_{a \in A} \in \mathcal{L}^A$. Let $x$ be $\delta$-inducible for some $(0,\beta)$-\deviation{} $\delta$ and let $z$ be an arbitrary feasible flow.
Then $C(x)/C(z) \leq (1 + \beta)/(1 - \hat{\mu}(\mathcal{L},\beta))$ if $\hat{\mu}(\mathcal{L},\beta) < 1$. Moreover, this bound is tight if $\mathcal{L}$ contains all constant functions and is closed under scalar multiplication, i.e., for every $l \in \mathcal{L}$ and $\gamma \geq 0$, $\gamma l \in \mathcal{L}$.
\end{rtheorem}

%\gsrem{added some comments back in here}

\bigskip
To see that the bound of Theorem \ref{thm:smoothness} is not worse than the bound $(1+\beta)/(1 - \mu)$, note that $(1,\mu)$-smooth latency functions we have 
$\hat{\mu}(\mathcal{L},\beta) \leq \hat{\mu}(\mathcal{L},0) \leq \mu$.

\begin{proof}[Theorem \ref{thm:smoothness}]
We use a similar approach as Correa et al. \cite{Correa2008}. Since $x$ is a \deviated{} Nash flow with respect to $l + \delta$, the following variational inequality holds:
$$
\sum_{a \in A} x_a(l_a(x_a) + \delta_a(x_a)) \leq \sum_{a \in A} z_a(l_a(x_a) + \delta_a(x_a)).
$$
We then have
\begin{align*}
C(x)  
& = \sum_{a \in A} x_al_a(x_a) 
\leq \sum_{a \in A} z_al_a(x_a) + (z_a - x_a)\delta_a(x_a)  \\
& \leq \sum_{x_a > z_a} z_al_a(x_a) + \sum_{z_a \geq x_a} z_a(l_a(x_a) + \delta_a(x_a))  \\
& \leq \sum_{x_a > z_a} z_al_a(x_a) + (1+\beta)\sum_{z_a \geq x_a} z_al_a(x_a) \\
& \leq \sum_{x_a > z_a} z_al_a(x_a) + (1+\beta)\sum_{z_a \geq x_a} z_al_a(z_a),
\end{align*}
where in the last inequality, we use that $x_a \leq z_a$ in the second summation. 

We obtain
\begin{eqnarray}
C(x)  &\leq & \sum_{x_a > z_a} z_al_a(x_a) + (1+\beta) \sum_{z_a \geq x_a} z_al_a(z_a) \nonumber \\
&=& \sum_{x_a > z_a} z_a[l_a(x_a) - (1+\beta)l_a(z_a) + (1+\beta)l_a(z_a)] + (1 + \beta)\sum_{z_a \geq x_a} z_al_a(z_a) \nonumber \\ 
&=& (1+\beta)C(z)  + \sum_{x_a > z_a} z_a[l_a(x_a) - (1+\beta)l_a(z_a)] \nonumber \\
&\leq& (1+\beta)C(z) + \hat{\mu}(\beta)  \sum_{x_a > z_a} x_al_a(x_a) \nonumber \\
& \leq & (1+\beta)C(z) +\hat{\mu}(\beta) C(x). \nonumber
\end{eqnarray}
Thus, for $\hat{\mu}(\beta) < 1$, we obtain $C(x)/C(z) \leq (1+\beta)/(1 - \hat{\mu}(\beta))$.

We will now prove the tightness of the obtained bound if $\mathcal{L}$ contains all constant functions and is closed under scalar multiplication. For a fixed $c = c(y) \in \mathcal{L}$, consider the parallel-arc instance in Figure \ref{fig:smoothness_tight} with fixed demand $r$.

\begin{figure}[h!]
\centering
\begin{tikzpicture}[
  ->,
  >=stealth',
  shorten >=0.5pt,
  auto,
  %node distance=2cm,
  semithick,
  every state/.style={circle,fill=white,minimum size=0.3cm,text=black},
]
\begin{scope}
  %Horizontal nodes
  \node[state]  (s)    {$s$};
  \node[state]  (t)   [right=4cm of s]  {$t$};

\path[every node/.style={sloped,anchor=south,auto=false}]
%Horizontal edges          
(s) edge  [bend right=20]  node[below=0.1cm] {$\left(\frac{1}{r},\beta \frac{1}{r}\right)$} (t)
(s) edge [bend left=20]  node[above=0.1cm] {$\left((1+\beta)\frac{c(y)}{rc(r)},0\right)$} (t);

\end{scope}
\end{tikzpicture}
\caption{Example used in the proof of Theorem~\ref{thm:smoothness}. The arcs are labeled by their respective $(l_a,\delta_a)$ functions. Note that $\delta \in \Delta(0,\beta)$.} \label{fig:smoothness_tight}
\end{figure}
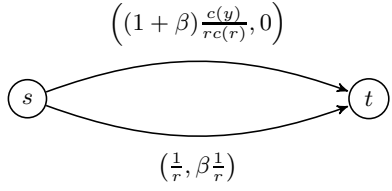 

Clearly, a \deviated{} Nash flow is given by $x = (x_1,x_2) = (r,0)$, since then $l_1(x_1) + \delta_1(x_1) = l_2(x_2) + \delta_2(x_2) = (1 + \beta)/r$. We have $C(x) = (1 + \beta)$.

Let the social optimum be given by $z^* = (\epsilon,r-\epsilon)$. We have 
$$
C(z^*) = \frac{(1+\beta)\epsilon c(\epsilon) + (r - \epsilon)c(r)}{rc(r)} = \frac{rc(r) - \epsilon[c(r) - (1+\beta)c(\epsilon)]}{rc(r)}
$$
which implies that
$$
\frac{C(x)}{C(z^*)}  = (1+\beta)\left(1 - \frac{\epsilon[c(r) - (1+\beta)c(\epsilon)]}{r\cdot c(r)}\right)^{-1}
$$
In order to claim tightness, let $c \in \mathcal{L}$ be such that it maximizes $\hat{\mu}(\mathcal{L},\beta)$, i.e., 
$$
\hat{\mu}(\mathcal{L},\beta) = \sup_{f, g \geq 0} \frac{f[c(g) - (1+\beta)c(f)]}{g\cdot c(g)}.
$$
Using this, we obtain
$$
\inf_{\epsilon, r \geq 0} \left\{ 1 - \frac{\epsilon[c(r) - (1+\beta)c(\epsilon)]}{r\cdot c(r)}\right\} = 1 - \sup_{\epsilon, r \geq 0} \frac{\epsilon[c(r) - (1+\beta)c(\epsilon)]}{r\cdot c(r)},
$$
which concludes the proof. 
\qed
\end{proof}

\subsection{Proposition~\ref{prop:mu-affine} and Corollary~\ref{cor:smoothness_approx}}

%\gsrem{put some commented material back in here}

\begin{proposition}\label{prop:mu-affine}
Let $\mathcal{L}$ be the set of all affine latency functions with non-negative coefficients. Then $\hat{\mu}(\mathcal{L},\beta) = 1/(4(1+\beta))$.
\end{proposition}
\begin{proof}%[$\hat{\mu}(\mathcal{L},\beta) = 1/(4(1+\beta))$ for affine latencies]
Let $l_a(y) = c_a y + d_a$ be an arbitrary affine latency function with $c_a, d_a \geq 0$. We need to show that
$$
z_a[c_a x_a + d_a - (1 + \beta)(c_a z_a + d_a)] \leq  1/(4(1+\beta))x_a[c_ax_a + d_a],
$$
or, equivalently,
$$
c_a[z_a x_a - (1+\beta)z_a^2] + d_a[z_a - z_a(1 + \beta)] \leq c_a[1/(4(1+\beta))x_a^2] + d_a[1/(4(1+\beta))x_a].
$$
It suffices to show that $z_a x_a - (1+\beta)z_a^2 \leq 1/(4(1+\beta))x_a^2$ and $z_a - z_a(1 + \beta) \leq 1/(4(1+\beta))x_a$. The second inequality is always true, using the non-negativity of $z_a,x_a$ and $\beta$. For the first inequality, we have
$$
0 \leq \left( \frac{x_a}{2} - (1 + \beta)z_a\right)^2 = (1 + \beta)^2z_a^2 + \frac{x_a^2}{4} - (1+\beta)x_az_a,
$$
which implies that 
$$
[1+\beta]\left(x_a z_a - (1 + \beta)z_a^2\right) \leq \frac{x_a^2}{4}.
$$
Dividing this inequality by $(1 + \beta)$ gives the desired result. Further, we have tightness for $(x_a,z_a) = \left(1,1/(2(1+\beta))\right)$. \qed
\end{proof}
%For the tightness, we consider a parallel arc instance with two arcs. Let $l_1(y) = 1/(1+\beta)$ with $\delta_1(y) = \beta l_1(y) = \beta/(1+\beta)$, and $l_2(y) = x$ with $\delta_2(y) = 0$. The \deviated{} Nash flow $x = f^\delta$ for the \deviated{} latencies $l + \delta$ is given by $(x_1,x_2) = (0,1)$ with social cost $C(x) = 1$. The social optimum is given by $z = (z_1,z_2) = (1 - 1/(2(1+\beta)), 1/2(1+\beta))$ with social cost 
%\begin{eqnarray}
%C(z) &=& z_1 l_1(z_1) + z_2l_2(z_2) \nonumber \\
%& =&  \left(1 - \frac{1}{2(1+\beta)}\right) \cdot \frac{1}{1+\beta} + \frac{1}{2(1+\beta)} \cdot \frac{1}{2(1+\beta)} \nonumber \\
%&=& \frac{2(1+\beta) - 1}{2(1+\beta)^2} + \frac{1}{4(1+\beta)^2} \nonumber \\
%&=& \frac{4(1+\beta) - 2 + 1}{4(1+\beta)^2} \nonumber \\
%&=& \frac{3 + 4\beta}{4(1+\beta)^2}.
%\end{eqnarray}
%This implies that
%$$
%\frac{C(x)}{C(z)} = \frac{4(1+\beta)^2}{3 + 4\beta} = \frac{4(1+\beta)^2}{4(1+\beta) - 1} = \frac{(1+\beta)}{1 - \frac{1}{4(1+\beta)}} = \frac{1+\beta}{1 - \hat{\mu}(\beta)}.
%$$
%Note that this instance is a generalization of Pigou's example.

\medskip

\begin{corollary}
Let $\mathcal{L}$ be a set of non-negative, non-decreasing and continuous functions (containing constants and closed under scalar multiplication). Let $\mathcal{G}$ be the set of all instances with $(l_a)_{a \in A} \in \mathcal{L}^A$. If $\hat{\mu}(\mathcal{L},\beta) < 1$, then
$$
|\text{BPoA}(\mathcal{G}, (0, \beta) ) - \dr(\mathcal{G}, (0, \beta))| \leq (1 + \beta)\frac{\hat{\mu}(\mathcal{L},\beta)}{1 - \hat{\mu}(\mathcal{L},\beta)}.
$$
\label{cor:smoothness_approx}
\end{corollary}

For example, for affine latencies we have $\hat{\mu}(\mathcal{L},\beta) = 1/(4(1+\beta))$ (see Proposition~\ref{prop:mu-affine}). As a result, $|\text{BPoA}(\mathcal{G}, (0, \beta) ) - \dr(\mathcal{G}, (0, \beta))| \leq 1/3$ for all $\beta \geq 0$. This implies that the gap is independent of the parameter $\beta$. This suggests that for large $\beta$ the Biased Price of Anarchy provides a good approximation for the Deviation Ratio (or the Price of Risk Aversion). Note that this does not follow from the bound $4(1 + \beta)/3$ for affine latencies obtained in \cite{Lianeas2015,Meir2015} (resulting from the upper bound $(1+\beta)/(1 - \mu)$ with $\mu = 1/4$).

\begin{proof}[Corollary \ref{cor:smoothness_approx}]
Consider the instance dedicated in Figure \ref{fig:smoothness_tight}.
The Nash flow with respect to $\delta = 0$ is given by $(c(r)/(1+\beta),1 - c(r)/(1+\beta))$ with social cost $C(z) = 1$. Further, as argued in the proof of Theorem \ref{thm:smoothness}, the deviated Nash flow $x$ has social cost $C(x) = 1 + \beta$. Thus,
$$
\frac{C(x)}{C(z)} = 1 + \beta 
\leq \dr(\mathcal{G}, (0, \beta)) 
\leq \text{BPoA}(\mathcal{G}, (0, \beta)) 
\leq \frac{1+\beta}{1 - \hat{\mu}(\mathcal{L},\beta)}.
$$
This implies that
\begin{align*}
|\text{BPoA}(\mathcal{G},(0,\beta)) - \dr(\mathcal{G},(0,\beta))| 
& \leq \frac{1+\beta}{1 - \hat{\mu}(\mathcal{L},\beta)} - (1 + \beta) \\
& = (1 + \beta)\bigg(\frac{1}{1 - \hat{\mu}(\mathcal{L},\beta)} - 1\bigg).
\end{align*}
\qed
\end{proof}

\subsection{General path deviations and proof of Theorem~\ref{thm:smooth_general}}

%\gsrem{put some commented material back in here}

\bigskip

\begin{rtheorem}{Theorem}{\ref{thm:smooth_general}}
Let $\mathcal{I}$ be a general multi-commodity instance with $(l_a)_{a \in A} \in \mathcal{L}^A$. Let $x$ be $\delta$-inducible with respect to some $(0,\beta)$-path \deviation{} $\delta$ and let $z$ an arbitrary feasible flow.
If $\hat{\mu}(\mathcal{L},0) < 1/(1 + \beta)$, then $C(x)/C(z) \leq (1 + \beta)/(1 - (1 + \beta)\hat{\mu}(\mathcal{L},0))$.
\end{rtheorem} 
\begin{proof}%[Theorem \ref{thm:smooth_general}]
We know that the flow $x$ satisfies the variational inequality
$$
\sum_{P \in \mathcal{P}} x_P [l_P(x) + \delta_P(x)] \leq \sum_{P \in \mathcal{P}} z_P [l_P(x) + \delta_P(x)].
$$
%which follows directly from the equilibrium conditions. 
It follows that
$$
C(x) \leq \sum_{P \in \mathcal{P}} x_P [l_P(x) + \delta_P(x)] \leq \sum_{P \in \mathcal{P}} z_P [l_P(x) + \delta_P(x)] \leq (1 + \beta)\sum_{P \in \mathcal{P}}z_Pl_P(x)
$$
using the non-negativity of flow and \deviations{}. Using the smoothness conditions, we find
$$
\sum_{P \in \mathcal{P}}z_Pl_P(x) = \sum_{a \in A}z_al_a(x_a) \leq \sum_{a \in A}\ z_a l_a(z_a) + \sum_{a \in A}\mu x_al_a(x_a) = C(z) + \mu C(x)
$$
from which the result follows. \qed
\end{proof} 

\end{document}